\documentclass[a4paper,UKenglish,cleveref, autoref, thm-restate]{lipics-v2019}

\usepackage{xspace}
\usepackage{upgreek}
\usepackage{color, colortbl}
\usepackage{comment}
\usepackage{booktabs}
\usepackage{cite}

\newcommand{\set}[1]{\left\{ #1 \right\}}
\newcommand{\rpar}[1]{\left( #1 \right)}

\newcommand{\integer}{\mathbb{Z}}
\renewcommand{\natural}{\mathbb{N}}

\DeclareMathOperator{\rank}{\mathrm{rank}}
\def\rnk#1#2#3{\rank(#1\mathbin{|}#2,#3)}
\DeclareMathOperator{\ranks}{\mathrm{ranks}}
\def\rnks#1#2#3{\ranks(#1\mathbin{|}#2,#3)}
\def\minrnk#1#2#3{\min(#1\mathbin{|}#2,#3)}
\def\maxrnk#1#2#3{\max(#1\mathbin{|}#2,#3)}
\def\Xrnk#1#2#3{X(#1\mathbin{|}#2,#3)}
\DeclareMathOperator{\minp}{\mathsf{Min}}
\DeclareMathOperator{\maxp}{\mathsf{Max}}
\def\decision#1#2{#1\{#2\}}

\DeclareMathOperator{\poly}{poly}
\DeclareMathOperator{\election}{\mathcal{E}}
\DeclareMathOperator{\tie}{\uptau}
\DeclareMathOperator{\tieprime}{\uptau^\prime}

\DeclareMathOperator{\ps}{\pi}

\def\eqdef{\mathrel{{:}{=}}}
\def\e#1{\emph{#1}}

\def\T{\mathbf{T}}
\def\P{\mathbf{P}}
\def\Q{\mathbf{Q}}
\def\M{\mathbf{M}}

\def\neccom{\mathsf{NecCom}}
\def\poscom{\mathsf{PosCom}}
\def\necmem{\mathsf{NecMem}}
\def\posmem{\mathsf{PosMem}}
\def\angs#1{\mathord{\langle{#1\rangle}}}
\def\neccompk{\neccom\angs{k}}
\def\poscompk{\poscom\angs{k}}
\def\necmempk{\necmem\angs{k}}
\def\posmempk{\posmem\angs{k}}

\def\NW{\text{$\mathsf{NW}$}\xspace}
\def\PW{\text{$\mathsf{PW}$}\xspace}
\def\NUW{\text{$\mathsf{NU}$}\xspace}
\def\PUW{\text{$\mathsf{PU}$}\xspace}

\def\bark{\mathord{\overline{k}}}
\def\barx#1{\mathord{\overline{#1}}}

\newcommand{\eat}[1]{}

\newcommand{\shortcite}[1]{\cite{#1}}

\def\bennydone#1{}

\title{Computing the Extremal Possible Ranks with Incomplete Preferences} 

\author{Aviram Imber}{Technion, Israel}{}{}{}
\author{Benny Kimelfeld}{Technion, Israel}{}{}{}

\authorrunning{A.~Imber and B.~Kimelfeld}
\Copyright{Aviram Imber, Benny Kimelfeld}

\ccsdesc[500]{Theory of computation~Algorithmic game theory and mechanism design}

\keywords{Positional scoring rules, Incomplete preferences.} 

\category{} 

\relatedversion{} 

\supplement{}

\nolinenumbers 

\begin{document}

\maketitle

\begin{abstract}
  Various voting rules are based on ranking the candidates by scores induced by aggregating voter preferences. A winner (respectively, unique winner) is a candidate who receives a score not smaller than (respectively, strictly greater than) the remaining candidates. Examples of such rules include the positional scoring rules and the Bucklin, Copeland, and Maximin rules. When voter preferences are known in an incomplete manner as partial orders, a candidate can be a possible/necessary winner based on the possibilities of completing the partial votes. Past research has studied in depth the computational problems of determining the possible and necessary winners and unique winners.

These problems are all special cases of reasoning about the range of possible positions of a candidate under different tiebreakers. We investigate the complexity of determining this range, and particularly the extremal positions. Among our results, we establish that finding each of the minimal and maximal positions is NP-hard for each of the above rules, including all positional scoring rules, pure or not. Hence, none of the tractable variants of necessary/possible winner determination remain tractable for extremal position determination. Tractability can be retained when reasoning about the top-$k$ positions for a fixed $k$. Yet, exceptional is Maximin where it is tractable to decide whether the maximal rank is $k$ for $k=1$ (necessary winning) but it becomes intractable for all $k>1$.
\end{abstract}

\section{Introduction}
A central task in social choice is that of winner determination---how to aggregate voter preferences to decide who wins. Relevant scenarios may be political elections, document rankings in search engines, hiring dynamics in the job market, decision making in multiagent systems, determination of outcomes in sports tournaments, and so on~\cite{DBLP:reference/choice/2016}. Different voting rules can be adopted for this task, and the computational social-choice community has investigated the algorithmic aspects of various specific instances of rules.
We focus here on rules that are based on ranking the candidates by scores induced by aggregating voter preferences. A prominent example is the family of the \e{positional scoring rules}: each voter assigns to each candidate a score based on the candidate’s position in the voter's ranking, and a winning candidate is one who receives the maximal sum of scores. Famous instantiations include the plurality rule (where a winner is most frequently ranked first), the veto rule (where a winner is least frequently ranked last), their generalizations to $t$-approval and $t$-veto, respectively, and the Borda rule (where the score is the position in the reverse order). There are also non-positional voting rules that are based on candidate scoring, such as
the Bucklin, Copeland, and Maximin rules.

The seminal work of Konczak and Lang~\shortcite{Konczak2005VotingPW} has addressed the situation where voter preferences are expressed or known in just a partial manner. More precisely, a partial voting profile consists of a partial order for each voter, and a completion consists of a linear extension for each of the partial orders. The framework gives rise to the computational problems of determining the \e{necessary winners} who win in every completion, and the \e{possible winners} who win in at least one completion. Each of these problems has two variants that correspond to two forms of winning: having a score not smaller than any other candidate (i.e., being a \e{co-winner}) and a having a score strictly greater than all other candidates (i.e., being the \e{unique winner}). These computational problems are challenging since, conceptually, they involve reasoning about the entire (exponential-size) space of completions. The complexity of these problems has been thoroughly studied in a series of publications that established the tractability of the necessary winners for positional scoring rules~\cite{DBLP:journals/jair/XiaC11}, and a full classification of a general class of positional scoring rules (the ``pure'' scoring rules) into tractable and intractable for the problem of the possible winners~\cite{DBLP:journals/jcss/BetzlerD10,DBLP:journals/jair/XiaC11,DBLP:journals/ipl/BaumeisterR12}.

Yet, the outcome of an election often goes beyond just reasoning about the maximal score. For example, the ranking among the other candidates might determine who will be the elected parliament members, the entries of the first page of the search engine, the job candidates to recruit, and the finalists of a sports competition. Studies on \e{social welfare}, for instance, concern the aggregation of voter preferences into a full ranking of the candidates~\cite{SATTERTHWAITE1975187,10.2307/41106036}. In the case of a positional scoring rule, the ranking order is determined by the sum of scores from voters and some tie-breaking mechanism~\cite{DBLP:journals/jair/MeirPRZ08}. When voter preferences are partial, a candidate can be ranked in different positions for every completion, and we can then reason about the range of these positions. In fact, the aforementioned computational problems can all be phrased as reasoning about the minimal and maximal ranks under different tiebreakers. A candidate $c$ is a \e{possible co-winner} if the minimal rank is one when the tiebreaker favors $c$ most, a \e{possible unique winner} if the minimal rank is one when the tiebreaker favors $c$ least, a \e{necessary co-winner} if the maximal rank is one when the tiebreaker favors $c$ most, and a \e{necessary unique winner} if the maximal rank is one when the tiebreaker favors $c$ least.

  {\begin{table*}[t]
      \def\arraystretch{1.05}     
      \caption{\label{tab:complexity} Overview of the results for positional scoring rules. $\bark$ stands for $m-k+1$ where $m$ is the number of candidates. Results on parameterized complexity take $k$ as the parameter.}
  \centering
  \scalebox{0.82}{
    \begin{tabular}{
      >{\centering\arraybackslash}m{2cm}|
      >{\centering\arraybackslash}m{2.7cm}
      >{\centering\arraybackslash}m{4cm}
      >{\centering\arraybackslash}m{3cm}|
      >{\centering\arraybackslash}m{3.1cm}
      }
      \hline
  \textbf{Problem} & plurality, veto & pure $-\set{\mbox{pl.},\mbox{veto}}$ & non-pure & comment \\
  \hline\hline
  $\decision{\minp}{<}$ & NP-c  \par $\mathrm{W}[2]$-hard for pl. & NP-c & NP-c  & NP-c: [Thm.~\ref{thm:minAllRules}] \par $\mathrm{W}[2]$: [Thm.~\ref{thm:plurality-w2-hard}]\\
 \hline
  $\decision{\maxp}{>}$ & NP-c \par $\mathrm{W}[1]$-hard & NP-c \par $\mathrm{W}[1]$-hard & NP-c \par $\mathrm{W}[1]$-hard & [Thm.~\ref{thm:maxAllRules}] \par [Thm.~\ref{thm:max-w1-hard}] \\
  \hline
  $\decision{\minp}{<k}$ & P & NP-c for strongly pure w/ poly.~scores & ? &  P: [Thm.~\ref{thm:ptwk-fixed-possible-veto-plurality}] \par NP-c: [Thm.~\ref{thm:PolyRulePosMemk}]\\
  \hline
  $\decision{\maxp}{> k}$ & P & P for poly.~scores & P for poly.~scores & [Thm.~\ref{thm:PolyRuleNecMemk}] \\
      \hline
      \vspace{0.2em}$\decision{\minp}{< \bark}$ & P & P for poly.~scores & P for poly.~scores & [Thm.~\ref{thm:poly-rule-bounded-reversed-min}]\\
  \hline
  $\decision{\maxp}{>\bark}$ &  P & NP-c for strongly pure bounded & ? &  P:  [Cor.~\ref{cor:max-fixed-reversed-veto-plurality}] \par NP-c: [Thm.~\ref{thm:max-bounded-reversed-strongly}] \\
  \hline
    \end{tabular}}
\end{table*}}
  
 We study the computational problems $\decision{\minp}{\theta}$ and $\decision{\maxp}{\theta}$, where $\theta$ is one of $<$ and $>$. The input consists of a partial profile, a candidate, a tie-breaking (total) order, and a number $k$, and the goal is to determine whether $x\theta k$ where $x$ is the minimal rank and the maximal rank, respectively, of the candidate. Our results are summarized in Table~\ref{tab:complexity}
 and Table~\ref{tab:complexityOther} for positional scoring rules and for other rules, respectively.
 (We exclude famous rules that are not naturally expressed as candidate scoring, e.g., Condorcet.)
 
As Table~\ref{tab:complexity} shows for positional scoring rules, determining the extremal ranks of a candidate is fundamentally harder than the $k=1$ counterparts (necessary/possible winners). For example, it is known that detecting the possible winners is NP-hard for every pure rule, with the exception of plurality and veto where the problem is solvable in polynomial time~\cite{DBLP:journals/jcss/BetzlerD10,DBLP:journals/jair/XiaC11,DBLP:journals/ipl/BaumeisterR12}. In contrast, we show that determining each of the minimum and maximum ranks is NP-hard for \e{every} positional scoring rule, pure or not, including plurality and veto. In particular, the tractability of the necessary winners \e{does not} extend to reasoning about the maximal rank.
The same goes for the Bucklin and Maximin rules, as can be seen in Table~\ref{tab:complexityOther}.

We also study the impact of fixing $k$ and consider the problems $\decision{\minp}{\theta k}$ and $\decision{\maxp}{\theta k}$ where the goal is to determine whether $x\theta k$ where, again, $x$ is the minimal/maximal rank. As shown in Table~\ref{tab:complexity}, we establish a more positive picture in the case of positional scoring rules: tractability for the maximum (assuming that the scores are polynomial in the number of candidates), and tractability of the minimum under plurality and veto. The degree of the polynomials depend on $k$, and we show that this is necessary (under standard assumptions of parameterized complexity) at least for the case of minimum, where the problem is $\mathrm{W[2]}$-hard for plurality, and for the case of maximum, where the problem is $\mathrm{W[1]}$-hard for every positional scoring rule. Tractability for the maximum is also retained for the non-positional Bucklin rule, as shown in Table~\ref{tab:complexityOther}.
Interestingly, Maximin behaves differently: while it is tractable to decide whether the maximal rank is at least $k$ for $k=1$ (i.e., the necessary-winner problem), it is NP-complete for all $k>1$.

The study of the range of possible ranks, beyond the very top, is related to the problem of \e{multi-winner election} that has been studied mostly in the context of \e{committee selection}. In that respect, our work can be viewed as reasoning about (necessary/possible) membership in the committee that consists of the $k$ highest ranked candidates. Yet, common studies consider richer notions of committee selection that look beyond the individual achievements of candidates. Indeed, various utilities have been studied for qualifying the selected committee, such as maximizing the number of voters with approved candidates~\cite{DBLP:conf/atal/AzizGGMMW15} and, in that spirit, the Condorcet committees~\cite{DBLP:conf/ijcai/ElkindLS11,DBLP:journals/mss/Darmann13}, aiming at a proportional representation via frameworks such as Chamberlin and Courant's~\shortcite{10.2307/1957270} and Monroe's~\shortcite{monroe_1995}, and the satisfaction of fairness and diversity constraints~\cite{DBLP:conf/ijcai/CelisHV18,DBLP:conf/aaai/BredereckFILS18}. Moreover, for some of the famous committee selection rules, determining the elected committee can be intractable even if voter preferences are complete~\cite{DBLP:conf/ijcai/ProcacciaRZ07,DBLP:journals/scw/ProcacciaRZ08,DBLP:journals/mss/Darmann13,DBLP:journals/tcs/SkowronYFE15}, in contrast to rank determination (which is always in polynomial time in the framework we adopt).

The problem of multi-winner determination for incomplete votes has been studied by Lu and Boutilier~\shortcite{DBLP:conf/ijcai/LuB13} in a perspective different from pure ranking: find a committee that minimizes the maximum objection (or ``regret'') over all possible completions.
\section{Preliminaries}

We begin with some notation and terminology. We focus on positional scoring rules, and we extend the definitions to other voting rules in Section~\ref{sec:other}.

\subsection{Voting Profiles and Positional Scoring}
Let $C = \set{c_1, \dots, c_m}$ be the set of \emph{candidates} (or \emph{alternatives}) and let $V = \set{v_1, \dots, v_n}$ be the set of \emph{voters}. A \emph{voting profile} $\T = (T_1, \dots, T_n)$ consists of $n$ linear orders on $C$, where each $T_i$ represents the ranking of $C$ by $v_i$.

A \e{positional scoring rule} $r$ is a series $\set{ \Vec{s}_m }_{m \in \natural^+}$ of $m$-dimensional score vectors $\Vec{s}_m = (\Vec{s}_m(1), \dots, \Vec{s}_m(m))$ of natural numbers where $\Vec{s}_m(1) \geq \dots \geq \Vec{s}_m(m)$ and $\Vec{s}_m(1) > \Vec{s}_m(m)$. We denote $\Vec{s}_m(j)$ by $r(m,j)$.  Some examples of positional scoring rules include the \emph{plurality} rule $(1, 0, \dots, 0)$, the \emph{$t$-approval} rule $(1, \dots, 1, 0, \dots, 0)$ that begins with $t$ ones, the \emph{veto} rule $(1, \dots, 1, 0)$, the \emph{$t$-veto} rule that ends with $t$ zeros, and the \emph{Borda} rule $(m-1, m-2, \dots, 0)$.

Given a voting profile $\T = (T_1, \dots, T_n)$, the score $s(T_i, c, r)$ that the voter $v_i$ contributes to the candidate $c$ is $r(m,j)$ where $j$ is the position of $c$ in $T_i$. The score of $c$ in $\T$ is $s(\T, c, r) = \sum_{i=1}^n s(T_i, c, r)$ or simply $s(\T, c)$ if $r$ is clear from context. A candidate $c$ is a \e{winner} (or \e{co-winner}) if $s(\T, c)\geq s(\T, c')$ for all candidates $c'$, and a \e{unique winner} if $s(\T, c)> s(\T, c')$ for for all candidates $c'\neq c$.

We make some conventional assumptions about the positional scoring rule $r$. We assume that $r(m,i)$ is computable in polynomial time in $m$, and the scores in each $\Vec{s}_m$ are co-prime (i.e., their greatest common divisor is one). A positional scoring rule is \emph{pure} if every $\Vec{s}_{m+1}$ is obtained from $\Vec{s}_m$ by inserting a score at some position.

\subsection{Partial Profiles}

A \e{partial voting profile} $\P = (P_1, \dots, P_n)$ consists of $n$ partial orders (i.e., reflexive, anti-symmetric and transitive relations) on the set $C$ of candidates, where each $P_i$ represents the incomplete preference of the voter $v_i$. A \e{completion} of $\P = (P_1, \dots, P_n)$ is a complete voting profile $\T = (T_1, \dots, T_n)$ where each $T_i$ is a completion (i.e., a linear extension) of the partial order $P_i$. The computational problems of determining the \e{necessary winners} and \e{possible winners} for partial voting preferences were introduced by Konczak and Lang~\cite{Konczak2005VotingPW}.

 Given
a partial voting profile $\P$, a candidate $c \in C$ is a
\emph{necessary winner} if $c$ is a winner in every completion $\T$ of
$\P$, and $c$ is a \emph{possible winner} if there exists a completion
$\T$ of $\P$ where $c$ is a winner. Similarly,  $c$ is a
\emph{necessary unique winner} if $c$ is a unique winner in every completion $\T$ of
$\P$, and $c$ is a \emph{possible unique winner} if there exists a completion
$\T$ of $\P$ where $c$ is a unique winner.

The decision problems associated
to a positional scoring rule $r$ are those of determining, given a
partial profile $\P$ and a candidate $c$, whether $c$ is a  necessary
winner, a necessary unique winner, a possible winner, and a possible unique winner. We denote these problems
by \NW, \NUW, \PW and \PUW, respectively. A known classification of the complexity of
these problems is the following. 
\begin{theorem}[Classification Theorem~\cite{DBLP:journals/jcss/BetzlerD10,DBLP:journals/jair/XiaC11,DBLP:journals/ipl/BaumeisterR12}]
  Each of \NW and \NUW can be solved in polynomial time for every positional scoring rule. Each of \PW and \PUW is solvable in polynomial time for plurality and veto; for all other pure scoring rules,  \PW and \PUW are NP-complete.
\label{thm:classification}
\end{theorem}
We aim at generalizing the Classification Theorem to determine the \e{minimal and maximal ranks}, as we formalize next.

\subsection{Minimal and Maximal Ranks}
The \e{rank} of a candidate is its position in the list of candidates, sorted by the sum of scores from the voters. However, for a precise definition, we need to resolve potential ties. Formally, let $r$ be a positional scoring rule, $C$ be a set of candidates, $\T$ a voting profile, and $\tie$ a \e{tiebreaker}, which is simply a linear order over $C$. Let $R_{\T}$ be the linear order on $C$ that sorts the candidates by their scores and then by $\tie$; that is,
$$ R_{\T} \eqdef \set{c \succ c' : s(\T, c) > s(\T, c') \lor (s(\T, c) = s(\T, c')\land c \tie c') }\,. $$
The rank of $c$ is the position of $c$ in $R_{\T}$, and we denote it
by $\rnk c\T\tie$.  If $\T$ is replaced with a partial voting profile
$\P$, then we define $\rnks c\P\tie$ as the set of ranks that $c$ gets
in the different completions of $\P$:
$$\rnks c\P\tie \eqdef \set{\rnk c\T\tie \mid \text{$\T$ extends
    $\P$}}$$
The minimal and maximal positions in $\rnks c\P\tie$ are denoted by $\minrnk c\P\tie$ and $\maxrnk c\P\tie$, respectively.    

Observe the following for a partial profile $\P$ and a candidate $c$:
\begin{itemize}
    \item $c$ is a possible winner if and only if $\minrnk c\P\tie=1$ (or $\minrnk c\P\tie < 2$) for any tiebreaker $\tie$ that positions $c$ first.
    \item $c$ is a possible unique winner if and only if $\minrnk c\P\tie=1$ for any tiebreaker $\tie$ that positions $c$ last.
    \item $c$ is a necessary winner if and only if $\maxrnk c\P\tie=1$ (or $\maxrnk c\P\tie < 2$) for any tiebreaker $\tie$ that positions $c$ first.
    \item $c$ is a necessary unique winner if and only if $\maxrnk c\P\tie=1$ for any tiebreaker $\tie$ that positions $c$ last.
\end{itemize}

To investigate the computational complexity of calculating the minimal and maximal ranks for a scoring rule $r$, we will consider the decision problems of determining, given $\P$, $c$, $\tie$ and a position $k$, whether $\Xrnk c\P\tie\mathrel{\theta} k$ where $X$ is one of $\min$ and $\max$ and $\theta$ is one of $<$ and $>$.  We denote these problems by $\decision{\minp_r}{\theta}$ and $\decision{\maxp_r}{\theta}$. Moreover, we will omit the rule $r$ when it is clear from the context. For example, $\decision{\minp_r}{<}$ (or just $\decision{\minp}{<}$) is the decision problem of determining whether $\minrnk c\P\tie<k$, and $\decision{\maxp_r}{>}$ (or just $\decision{\maxp}{>}$) decides whether $\maxrnk c\P\tie>k$.

Observe that for every scoring rule $r$, if we can compute the scores of the candidates within a complete profile in polynomial time, then $\decision{\minp}{<}$ and $\decision{\maxp}{>}$ are in NP. Also observe that if $\decision{\minp}{<}$ is solvable in polynomial time, then so is $\decision{\minp}{>}$.  Conversely, if $\decision{\minp}{<}$ is NP-complete then $\decision{\minp}{>}$ is coNP-complete. The same holds for the complexity of $\decision{\maxp}{>}$ in comparison to $\decision{\maxp}{<}$. Hence, in the remainder of the paper we will restrict the discussion to
$\decision{\minp}{<}$ and $\decision{\maxp}{>}$.

\subsection{Additional Notation.}
For a set $A$ and
a partition $A_1, \dots, A_t$ of $A$, $P(A_1, \dots, A_t)$ denotes the partitioned partial order $\set{a_1 \succ \dots \succ a_t : \forall i \in [t], a_i \in A_i}$ and $O(A_1, \dots, A_t)$ denotes an arbitrary linear order on $A$ that completes $P(A_1, \dots, A_t)$. A linear order $a_1 \succ \dots \succ a_t$ is also denoted as a vector $(a_1, \dots, a_t)$. The \e{concatenation} 
  $(a_1, \dots, a_t) \circ (b_1, \dots, b_\ell)$ is $(a_1, \dots, a_t, b_1,\dots, b_\ell)$.

\section{Positional Scoring Rules}

In this section, we show that the problems we study are computationally
hard for \e{all} positional scoring rules. The following theorems
state the hardness of computing the minimal and maximal rank for all positional scoring rules.

\renewcommand{\arraystretch}{1.2}

\begin{figure}[t]
\centering
\newcolumntype{g}{>{\columncolor{lipicsLightGray}}c}
\begin{tabular}{|c|c|c|c|c|g|c|c|c|}
  \hline
  Voter & 1 & 2 & $\cdots$ & $\ell-1$ & $\ell, \ell+1$ & $\ell+2$ & $\cdots$ & $n+2$
  \\
  \hline
  $P^1_e(i)$ & $c_1$ & $c_2$ & $\cdots$ & $c_{\ell-1}$ & $\set{u, w}$ & $c_\ell$ & $\cdots$ & $c_n$ \\
  \hline
  $P^1_e(2)$ & $c_2$ & $c_3$ & $\cdots$ & $c_\ell$ &  $\set{u, w}$ & $c_{\ell+1}$ & $\cdots$ & $c_1$ \\
  \hline
  $\vdots$ & & & & & & & & \\
  \hline
  $P^1_e(n)$ & $c_n$ & $c_1$ & $\cdots$ & $c_{\ell-2}$ &  $\set{u, w}$ & $c_{\ell-1}$ & $\cdots$ & $c_{n-1}$ \\
  \hline
\end{tabular}
\caption{\label{fig:min-all-rules-Pe}The voters of the profile $\P^1_e = (P^1_e(1), \dots, P^1_e(n))$ for the edge $e = \set{u, w}$ used in the proof of Theorem~\ref{thm:minAllRules}. The other candidates are denoted as $C \setminus U = \set{c_1, \dots, c_n}$.}
\end{figure}

\begin{figure}[b]
\centering
\newcolumntype{g}{>{\columncolor{lipicsLightGray}}c}
\begin{tabular}{|c|c|c|c|c|g|g|c|c|c|}\hline
  Voter & 1 & 2 & $\cdots$ & $\ell-1$ & $\ell$ & $\ell+1$ & $\ell+2$ & $\cdots$ & $n+2$
  \\
  \hline
  $T^2_1$ & $u_1$ & $u_2$ & $\cdots$ & $u_{\ell-1}$ & $d$ & $c^*$ & $u_\ell$ & $\cdots$ & $u_n$ \\
  \hline
  $T^2_2$ & $u_2$ & $u_3$ & $\cdots$ & $u_\ell$ & $d$ & $c^*$ & $u_{\ell+1}$ & $\cdots$ & $u_1$ \\
  \hline
  $\vdots$ & & & & & & & & & \\
  \hline
  $T^2_n$ & $u_n$ & $u_1$ & $\cdots$ & $u_{\ell-2}$ & $d$ & $c^*$ & $u_{\ell-1}$ & $\cdots$ & $u_{n-1}$ \\
  \hline
\end{tabular}
\caption{\label{fig:minAllRules-T}The voters of the profile $(T_1^2, \dots, T_n^2)$ used in the proof of Theorem~\ref{thm:minAllRules}.}
\end{figure}

\def\thmminAllRules{ For every positional scoring rule $r$, $\decision{\minp_r}{<}$ is NP-complete. }
\begin{theorem}
\label{thm:minAllRules}
\thmminAllRules
\end{theorem}
\begin{proof}
Let $r=\set{\Vec{s}_m}_{m > 1}$ be a positional scoring rule. We assume, without loss of generality, that $\Vec{s}_m (m) = 0$ for every $m > 1$. (Otherwise, we can subtract $\Vec{s}_m (m)$ from all the entries in the vector without affecting the ranks in any profile.)  The membership of $\decision{\minp_r}{<}$ in NP is straightforward. We show hardness by a reduction from the vertex-cover problem: given an undirected graph $G$ and an integer $k$, is there a set $B$ of $k$ or fewer vertices such that every edge is incident to at least one vertex in $B$? This problem is known to be NP-complete even on regular graphs~\cite{DBLP:journals/tcs/GareyJS76}, and we will assume that $G$ is indeed regular.

Let $G = (U, E)$ be a regular graph with $U = \set{u_1, \dots, u_n}$. In the reduction, the vertices will correspond to candidates, and the edges will be voters that will need to select one of their incident vertices.
Hence, the edges jointly select a vertex cover.
The question will be whether this vertex cover is small enough. The details follow.

We construct an instance $(C, \P, \tie)$ under $r$. The candidate set is $C = U \cup \set{c^*, d}$ and the tiebreaker is $\tie = O(\set{c^*, d}, U)$. The voting profile $\P = \P^1 \circ \T^2$ is the concatenation of two parts $\P^1$ and $\T^2$ that we describe next. Note that $|C| = n+2$. Let $\ell < n+2$ be an index where $\Vec{s}_{n+2}(\ell) > \Vec{s}_{n+2}(\ell+1) = 0$.
We know that such $\ell$ exists due to the definition of a scoring rule and our assumption that
$\Vec{s}_m (m) = 0$ for every $m > 1$.

The first part of the profile contains a profile for every edge $\P^1 = \set{\P^1_e}_{e \in E}$. For every edge $e = \set{u,w} \in E$, the profile $\P^1_e = (P^1_e(1), \dots, P^1_e(n))$ consists of $n$ voters, as illustrated in Figure~\ref{fig:min-all-rules-Pe}. For every $i \in [n]$, denote $M_i(C \setminus e) = (c_{i_1}, \dots, c_{i_n})$ where $M_i$ is the \e{$i$th circular vote} as defined by Baumeister, Roos and J{\"{o}}rg~\cite{DBLP:conf/atal/BaumeisterRR11}: 
$$ M_i(a_1, \dots, a_t) \eqdef (a_i, a_{i+1}, \dots, a_t, a_1, a_2, \dots, a_{i-1})\,. $$
Then $P^1_e(i) \eqdef \rpar{c_{i_1}, c_{i_2}, \dots, c_{i_{\ell-1}}, \set{u,w}, c_{i_\ell}, c_{i_{t+1}}, \dots, c_{i_n}}$ is the $i$th voter in $\P^1_e$. This means that in $\P^1_e$, the candidates $u$ and $w$ can only be at positions $\ell$ and $\ell+1$, and the other candidates are circulating at all other positions. The decision whether to rank $u$ or $w$ at the $\ell$th position represents the selection of $e$ between its vertices, to construct a vertex cover.

The second part of the profile, $\T^2$, is constructed such that for every completion $\T$ of $\P$ and vertex $u \in U$, the candidate $c^*$ defeats $u$ if and only if for every edge $e$ incident to $u$, all voters of $\P_e^1$ rank $u$ at the $(\ell+1)$st position. (This means that none of the edges have selected $u$, and so $u$ is not in the constructed cover.) Formally, recall that $G$ is regular, and
let $\Delta$ be the common degree of all the vertices of $G$. The profile
$\T^2$ consists of $\Delta$ copies of $(T^2_1, \dots, T^2_n)$, as illustrated in Figure~\ref{fig:minAllRules-T}. For every $i \in [n]$, denote $M_i(U) = (c_{i_1}, \dots, c_{i_n})$ and define
 $T^2_i = (c_{i_1}, c_{i_2}, \dots, c_{i_{\ell-1}}, d, c^*, c_{i_\ell}, c_{i_{t+1}}, \dots, c_{i_n})$.
This means that $d$ and $c^*$ are always at positions $\ell$ and $ \ell+1$, respectively, and the candidates of $U$ are circulating at all other positions. This completes the construction of $(C, \P, \tie)$.

For the correctness, we start with some observations regrading the voting profile. Let $\T = \T^1 \circ \T^2$ be a completion of $\P$ where $\T^1 = \set{\T^1_e}_{e \in E}$. The scores of $c^*$ and $d$ in $\T^1$ are
\begin{align*}
    s(\T^1, c^*) = s(\T^1, d) = \sum_{e \in E} s(\T^1_e, d) = |E| \sum_{i \neq \ell, \ell+1} \Vec{s}_{n+2}(i)\,.
\end{align*}
For every $u \in U$, denote by $E(u)$ the set of edges incident to $u$, and denote $\overline{E}(u) = E \setminus E(u)$. Recall that $|E(u)| = \Delta$ since the graph is regular. By definition, it holds that
\[s(\T^1, u) = \sum_{e \in E(u)} s(\T^1_e, u) + \sum_{e \in \overline{E}(u)} s(\T^1_e, u)\,.\]
Observe that 
\[\sum_{e \in E(u)} s(\T^1_e, u) \leq \Delta n \cdot \Vec{s}_{n+2}(\ell)\] and that
\[\sum_{e \in \overline{E}(u)} s(\T^1_e, u) = (|E|-\Delta) \sum_{i \neq \ell, \ell+1} \Vec{s}_{n+2}(i)\,.\]

In $\T^2$, the score of $c^*$ is $s(\T^2, c^*) = \Delta n \cdot \Vec{s}_{n+2}(\ell+1) = 0$ and the score of $d$ is $s(\T^2, d) = \Delta n \cdot \Vec{s}_{n+2}(\ell)$. For every $u \in U$, $s(\T^2, u) =  
\Delta \sum_{i \neq \ell, \ell+1} \Vec{s}_{n+2}(i)$.

Overall,
\begin{align}
    s(\T, c^*) &= s(\T^1, c^*) + s(\T^2, c^*) = |E| \sum_{i \neq \ell, \ell+1} \Vec{s}_{n+2}(i) \label{eq:cstar}\\
    s(\T, d) &= s(\T^1, d) + s(\T^2, d) = |E| \sum_{i \neq \ell, \ell+1} \Vec{s}_{n+2}(i) + \Delta n \cdot \Vec{s}_{n+2}(\ell) \label{eq:d}
\end{align}
and for every $u \in U$,
\begin{align}
    s(\T, u) &= s(\T^1, u) + s(\T^2, u) \notag \\
    &= \sum_{e \in E(u)} s(\T^1_e, u) + \sum_{e \in \overline{E}(u)} s(\T^1_e, u) + s(\T^2, u) \notag\\
    &= \sum_{e \in E(u)} s(\T^1_e, u) + (|E|-\Delta) \sum_{i \neq \ell, \ell+1} \Vec{s}_{n+2}(i) + \Delta \sum_{i \neq \ell, \ell+1} \Vec{s}_{n+2}(i) \notag\\
    &= \sum_{e \in E(u)} s(\T^1_e, u) + |E| \sum_{i \neq \ell, \ell+1} \Vec{s}_{n+2}(i) \label{eq:equ}\\
    &\leq \Delta n \cdot \Vec{s}_{n+2}(\ell) + |E| \sum_{i \neq \ell, \ell+1} \Vec{s}_{n+2}(i) = s(\T, d)\label{eq:ltu}
\end{align}
From Equations~\eqref{eq:d} and~\eqref{eq:ltu} and the definition of $\tie$ we conclude that $d$ always defeats all other candidates.
From Equations~\eqref{eq:cstar} and~\eqref{eq:equ} we conclude that
$c^*$ defeats $u$ if and only if $\sum_{e \in E(u)} s(\T^1_e, u) = 0$.

Recall that $\alpha(G)$ denotes the minimal size of a vertex cover in $G$. We now show that for any $k$, $\alpha(G) \leq k$ if and only if $\minrnk {c^*}\P\tie \leq k+2$, which implies NP-completeness for $\decision{\minp_r}{<}$.

Assume that $\alpha(G) \leq k$, and let $B$ be a vertex cover of size at most $k$ in $G$. Consider the following completion $\T = \set{\T^1_e}_{e \in E} \circ \T^2$ of $\P$. For every $e = \set{u, w} \in E$, recall that only the positions of $u$ and $w$ are not determined in the voters of $\P^1_e$. If $u \in B$, then in all voters of $\T^1_e$ the candidate $u$ is placed at the $\ell$th position and $w$ is placed at the $(\ell+1)$th position. Otherwise, $w \in B$ (since $B$ is a vertex cover), and then in all voters of $\T^1_e$ the candidate $w$ is placed at the $\ell$th position and $u$ is placed at the $(\ell+1)$th position.

So, for every $u \notin B$ the candidate $u$ is placed at the $(\ell+1)$th position in all voters of $\set{\T^1_e}_{e \in E(u)}$, hence $\sum_{e \in E(u)} s(\T^1_e, u) = 0$ and $c^*$ defeats $u$. These are at least $n-k$ candidates which $c^*$ defeats, therefore $\rnk {c^*}\T\tie \leq k+2$ and $\minrnk {c^*}\P\tie \leq k + 2$.

Conversely, assume that  $\minrnk {c^*}\P\tie \leq k + 2$, and let  $\T = \set{\T^1_e}_{e \in E} \circ \T^2$ be a completion of $\P$ where $\rnk {c^*}\T\tie \leq k+2$. Let $B \subseteq U$ be the candidates of $U$ that defeat $c^*$ in $\T$, we know that $|B| \leq k$ because $d$ always defeats $c^*$. For every $e = \set{u, w} \in E$, all voters of $\T^1_e$ placed a vertex from $B$ at the $\ell$th position. (If a vertex from $U \setminus B$ is placed at the $\ell$th position for some voter, then this vertex defeats $c^*$, in contradiction to the definition of $B$.) Since these voters can only place $u$ and $ w$ at the $\ell$th position, we get that either $u \in B$ or $w \in B$. Hence $B$ is a vertex cover, which implies $\alpha(G) \leq k$. We conclude the correctness of the reduction and, hence, the NP-completeness of $\decision{\minp_r}{<}$.
\end{proof}

{\begin{figure}
   \small
   \centering
   \scalebox{0.9}{
\begin{tabular}{|c|c|c|c|c|c|c|c|c|c|c|c|}\hline
  Voter & 1 & 2 & $\cdots$ & $\ell-1$ & $\ell$ & $\ell+1$ & $\ell+2$ & $\cdots$ & $n$ & $n+1$ & $n+2$
  \\
  \hline
  $T^3_1$ & $u_1$ & $u_2$ & $\cdots$ & $u_{\ell-1}$ & $u_\ell$ & $u_{\ell+1}$ & $u_{\ell+2}$ & $\cdots$ & $u_n$ & \cellcolor{lipicsLightGray} $d$ & \cellcolor{lipicsLightGray} $c^*$ \\
  \hline
  $T^3_2$ & $u_2$ & $u_3$ & $\cdots$ & $u_{\ell}$ & $u_{\ell+1}$ & $u_{\ell+2}$ & $u_{\ell+3}$ & $\cdots$ & \cellcolor{lipicsLightGray} $d$ & \cellcolor{lipicsLightGray} $c^*$ & $u_1$ \\
  \hline
  $\vdots$ & & & & & & & & & & & \\
  \hline
  $T^3_{n-\ell+1}$ & $u_{n-\ell+1}$ & $u_{n-\ell+2}$ & $\cdots$ & $u_{n-1}$ & $u_{n}$ & \cellcolor{lipicsLightGray} $d$ & \cellcolor{lipicsLightGray} $c^*$ & $\cdots$ & $u_{n-\ell-2}$ & $u_{n-\ell-1}$& $u_{n-\ell}$ \\
  \hline
  $T^3_{n-\ell+2}$ & $u_{n-\ell+2}$ & $u_{n-\ell+3}$ & $\cdots$ & $u_{n}$ & \cellcolor{lipicsYellow} $c^*$ & \cellcolor{lipicsYellow} $d$ & $u_1$ & $\cdots$ & $u_{n-\ell-1}$ & $u_{n-\ell}$& $u_{n-\ell+1}$ \\
  \hline
  $T^3_{n-\ell+3}$ & $u_{n-\ell+3}$ & $u_{n-\ell+4}$ & $\cdots$ & \cellcolor{lipicsLightGray} $d$ & \cellcolor{lipicsLightGray} $c^*$ & $u_1$ & $u_2$ & $\cdots$ & $u_{n-\ell}$ & $u_{n-\ell+1}$& $u_{n-\ell+2}$ \\
  \hline
  $\vdots$ & & & & & & & & & & & \\
  \hline
  $T^3_{n+2}$ & \cellcolor{lipicsLightGray} $c^*$ & $u_{1}$ & $\cdots$ & $u_{\ell-2}$ & $u_{\ell-1}$ & $u_{\ell}$ & $u_{\ell+1}$ & $\cdots$ & $u_{n-1}$ & $u_{n}$& \cellcolor{lipicsLightGray} $d$ \\
  \hline
\end{tabular}}
\caption{\label{fig:max-T3i}The voters of the profile $(T_1^3, \dots, T_{n+2}^3)$ used in the proof of Theorem~\ref{thm:maxAllRules}.} 
\end{figure}}

Theorem~\ref{thm:minAllRules} stated the hardness of $\decision{\minp_r}{<}$ for every positional scoring rule $r$.  The next theorem states the hardness of $\decision{\maxp_r}{>}$ for every such $r$.

\def\thmmaxAllRules{ For every positional scoring rule $r$, $\decision{\maxp_r}{>}$ is NP-complete. }
\begin{theorem}
\label{thm:maxAllRules}
\thmmaxAllRules
\end{theorem}
\begin{proof}
This proof uses parts of the proof of Theorem~\ref{thm:minAllRules}.  Let $r=\set{\Vec{s}_m}_{m > 1}$ be a positional scoring rule. We again assume (w.l.o.g.) that $\Vec{s}_m (m) = 0$ for every $m > 1$. Membership of $\decision{\maxp_r}{>}$ in NP is straightforward.  We show hardness by a reduction from the independent-set problem: given an undirected graph $G$ and an integer $k$, is there any set $B \subseteq U$ of $k$ or more vertices such that no two vertices in $B$ are connected by an edge? Again, we use the NP-complete variant of the problem where $G$ is regular~\cite{DBLP:journals/tcs/GareyJS76}.

Let $G = (U, E)$ be a regular graph with $U = \set{u_1, \dots, u_n}$, and let $\Delta$ be the degree of all vertices. As in the proof of Theorem~\ref{thm:minAllRules}, we will make every edge (voter) select an incident vertex (candidate).
Let $B \subseteq U$ be the vertices who receive $\Delta$ votes. Observe that $B$ is necessarily an independent set. The question is whether we can construct a big enough such $B$. Details follow.

We construct an instance $(C, \P, \tie)$ under $r$, as follows. The candidates set is $C = U \cup \set{c^*, d}$ and the tiebreaker is $\tie = O(U, \set{c^*, d})$.
Note that $|C| = n+2$.
The voting profile is the concatenation $\P = \P^1 \circ \T^2 \circ \T^3$ of three parts described next.

Let $\ell < n+2$ be an index such that $\Vec{s}_{n+2}(\ell) > \Vec{s}_{n+2}(\ell+1) = 0$. The first two parts $\P^1 = \set{\P^1_e}_{e \in E}$ and $\T^2$ are the same as in the proof of Theorem~\ref{thm:minAllRules}. Recall that for every edge $e = \set{u, w}$, only the positions of $u$ and $w$ are not determined in the voters of $\P_e^1$. The edge $e$ ``selects'' the vertex that is put in the $\ell$th position.

The third part, $\T^3$, is constructed such that for every completion $\T$ of $\P$ and vertex $u \in U$ it holds that $u$ defeats $c^*$ if and only if
all voters of $\P_e^1$ rank $u$ at the $\ell$th position for every edge $e$ incident to $u$. (This means that all edges incident to $u$ select $u$.) Formally, $\T^3$ consists of $\Delta n$ copies of the profile $(T^3_1, \dots, T^3_{n+2})$, as illustrated in Figure~\ref{fig:max-T3i}. We start with $T^3_i = M_i(u_1, \dots, u_n, d, c^*)$ for the circular votes as defined in the proof of Theorem~\ref{thm:minAllRules}, and then perform the following change. There exists some $i \in [n+2]$ such that $d$ and $c^*$ are placed at positions $\ell$ and $\ell+1$, respectively, in $T^3_i$. In this voter, switch the positions of $d$ and $c^*$. This means that in $(T^3_1, \dots, T^3_{n+2})$, the candidate $c^*$ is placed at the $\ell$th position twice, and $d$ is placed at the $(\ell+1)$st position twice.

Now, we prove the correctness. Since $\Vec{s}_{n+2} (\ell+1) = 0$, observe that 
$s(\T^3, c^*) = \Delta n \rpar{\sum_{i=1}^{n+2} \Vec{s}_{n+2} (i) + \Vec{s}_{n+2}(\ell)}$
and $s(\T^3, d) = \Delta n \rpar{\sum_{i=1}^{n+2} \Vec{s}_{n+2} (i) - \Vec{s}_{n+2}(\ell)}$. For every $u \in U$ we have that $s(\T^3, u) = \Delta n \sum_{i=1}^{n+2} \Vec{s}_{n+2} (i)$. By combining this with the observations from the proof of Theorem~\ref{thm:minAllRules}, we get that for every completion $\T = \set{\T^1_e}_{e \in E} \circ \T^2 \circ \T^3$ of $\P$, the following holds.
The score of $c^*$ is given by
\begin{align*}
    s(\T, c^*) &= (s(\T^1, c^*) + s(\T^2, c^*)) + s(\T^3, c^*) \\
    &= |E| \sum_{i \neq \ell, \ell+1} \Vec{s}_{n+2}(i) + \Delta n \sum_{i=1}^{n+2} \Vec{s}_{n+2}(i) + \Delta n \cdot \Vec{s}_{n+2}(\ell)\,.
\end{align*}
The score of $d$ is given by
\begin{align*}
    s(\T, d) &= (s(\T^1, d) + s(\T^2, d)) + s(\T^3, d) \\
    &= \rpar{|E| \sum_{i \neq \ell, \ell+1} \Vec{s}_{n+2}(i) + \Delta n \cdot \Vec{s}_{n+2}(\ell)} + \Delta n \rpar{ \sum_{i=1}^{n+2} \Vec{s}_{n+2}(i) - \Vec{s}_{n+2}(\ell)}\\
    &= |E| \sum_{i \neq \ell, \ell+1} \Vec{s}_{n+2}(i) + \Delta n \sum_{i=1}^{n+2} \Vec{s}_{n+2}(i)\,.
\end{align*}
The score of every $u \in U$ is given by
\begin{align*}
    s(\T, u) &= \sum_{e \in E(u)} s(\T^1_e, u) + \rpar{\sum_{e \in \overline{E}(u)} s(\T^1_e, u) + s(\T^2, u) + s(\T^3, u)} \\
    &= \sum_{e \in E(u)} s(\T^1_e, u) + |E| \sum_{i \neq \ell, \ell+1} \Vec{s}_{n+2}(i) + \Delta n \sum_{i=1}^{n+2} \Vec{s}_{n+2}(i)\\
    &\leq \Delta n \cdot \Vec{s}_{n+2}(\ell) + |E| \sum_{i \neq \ell, \ell+1} \Vec{s}_{n+2}(i) + \Delta n \sum_{i=1}^{n+2} \Vec{s}_{n+2}(i)\\
    &= s(\T, c^*)\,.
\end{align*}
By this analysis and the definition of $\tie$, $d$ is always defeated by all other candidates, and $u$ defeats $c^*$ if and only if $\sum_{e \in E(u)} s(\T^1_e, u) = \Delta n \cdot \Vec{s}_{n+2}(\ell)$.

Recall that $\beta(G)$ denotes the maximal size of an independent set of $G$.
We show that for
any $k$, $\beta(G) \geq k$ if and only if $\maxrnk {c^*}\P\tie \geq k+1$, thus $\decision{\maxp_r}{>}$ is NP-complete.

Assume that $\beta(G) \geq k$, let $B$ be an independent set of size at least $k$ in $G$. Consider a completion  $\T = \set{\T^1_e}_{e \in E} \circ \T^2 \circ \T^3$ of $\P$ as follows. For every $e = \set{u, w} \in E$, recall that only the positions of $u,w$ are not determined in the voters of $\P^1_e$. If $u, w \notin B$ then complete all voters of $\P^1_e$ arbitrarily.  If $u \in B$ then in all voters of $\T^1_e$, $u$ is placed at the $\ell$th position and $w$ is placed at the $(\ell+1)$th position. Finally, if $w \in B$ then in all voters of $\T^1_e$, $w$ is placed at the $\ell$th position and $u$ is placed at the $(\ell+1)$th position. Note that we cannot have $u, w \in B$ because $B$ is an independent set.

For every $u \in B$, $u$ is placed at the  $\ell$th position in all voters of $\set{\T^1_e}_{e \in E(u)}$, hence $\sum_{e \in E(u)} s(\T^1_e, u) = \Delta n \cdot \Vec{s}_{n+2}(\ell)$ and $u$ defeats $c^*$. These are at least $k$ candidates which defeat $c^*$, therefore $\rnk {c^*}\T\tie \geq k+1$ and $\maxrnk {c^*}\P\tie \geq k+1$.

Conversely, assume that $\maxrnk {c^*}\P\tie \geq k + 1$, there exists a completion $\T = \set{\T^1_e}_{e \in E} \circ \T^2 \circ \T^3$ of $\P$ where $\rnk {c^*}\T\tie \geq k+1$. Let $B$ be the candidates which defeat $c^*$ in $\T$. $B \subseteq U$ because $c^*$ always defeats $d$, and $|B| \geq k$. For every $u \in B$ we get that $\sum_{e \in E(u)} s(\T^1_e, u) = \Delta n \cdot \Vec{s}_{n+2}(\ell)$, otherwise $u$ does not defeat $c^*$.

Assume to the contrary that  $e = \set{u, w} \in E$ for some pair $u, w \in B$. In the voters of $\T^1_e$, both $u$ and $w$ should be placed at the $\ell$th position, that is a contradiction. Hence $B$ is an independent set, which implies $\beta(G) \geq k$.
\end{proof}
\subsection{Comparison to a Bounded Rank}
In the previous section, we established that the problems of computing
the minimal and maximal ranks are very often intractable. We now
investigate the complexity of comparing the minimal and maximal ranks
to some fixed rank $k$. Hence, the input consists of only $\P$, $c$ and $\tie$, but not
$k$. We denote these problems by $\decision{\minp_r}{{>}k}$ and
$\decision{\maxp_r}{{<}k}$. Again, we will omit the rule $r$ when it is clear from
the context. For example, $\decision{\minp}{{<}k}$ is the decision
problem of determining whether $\minrnk c\P\tie<k$.

We will show that the complexity picture for $\decision{\minp}{{<}k}$
and $\decision{\maxp}{{>}k}$ is way more positive, as we generalize
the tractability of \e{almost} all tractable scoring rules for
\NW and \PW. We will also generalize hardness results from \PW to
$\decision{\minp}{{<}k}$; interestingly, this generalization turns out
to be nontrivial.

In addition to comparing to the fixed $k$, we will consider the
problem of comparing to $\bark \eqdef m-k+1$ where $m$ is, as usual,
the number of candidates. Note that the
position $\bark$ is the $k$th rank from the end (bottom).  For instance, $\decision{\maxp}{{>}\barx{4}}$ decides whether the candidate can end up in one of the bottom $3$ positions.

\subsubsection{Complexity of $\decision{\minp}{{<}k}$}

We first show that the positional scoring rules that are tractable for
\PW, namely plurality and veto, are also tractable for
$\decision{\minp}{{<}k}$.  This is proved via a reduction to the
problem of \e{polygamous
  matching}~\cite{DBLP:conf/pods/KimelfeldKT19}: given a bipartite
graph $G = (U \cup W, E)$ and natural numbers $\alpha_w \leq \beta_w$
for all $w \in W$, determine whether there is a subset of $E$ where
each $u \in U$ is incident to exactly one edge and every $w \in W$ is
incident to at least $\alpha_w$ edges and at most $\beta_w$ edges.
This problem is known to be solvable in polynomial time.

\def\lemmacommitteeMatching{ The following decision problem can be solved in polynomial time for
the plurality and veto rules: given a partial profile $\P$ over a set $C$ of candidates and numbers $\gamma_c \leq \delta_c$ for every candidate $c$, is there a completion $\T$ such that $\gamma_c \leq s(\T, c) \leq \delta_c$ for every $c \in C$? }
\begin{lemma}
\label{lemma:committeeMatching}
\lemmacommitteeMatching
\end{lemma}
\begin{proof}
For both rules, we apply a reduction to polygamous matching, where $U=V$ (the set of voters) and $W = C$. For plurality, $E$ connects $v_i \in V$ and $c \in C$ whenever $c$ can be in the top position in one or more completions of $P_i$, and the bounds are $\alpha_c=\gamma_c$ and $\beta_c=\delta_c$. For veto, receiving a score $s$ is equivalent to being placed in the bottom position of $n-s$ voters, so $E$ connects $v_i \in V$ and $c \in C$ whenever $c$ can be in the bottom position in one or more completions of $P_i$. The bounds are $\alpha_c = n-\delta_c$ and $\beta_c = n-\gamma_c$.
\end{proof}

To solve $\decision{\minp}{{<}k}$ given $C$, $\P$, $\tie$ and $c$, we search for a completion where $c$ defeats more than $m-k$ candidates. For this goal we consider
every set $D \subseteq C \setminus \set{c}$ of size $m-k+1$ and search
for a completion where $c$ defeats all candidates of $D$. For that, we
iterate over every integer score $0\leq s\leq n$ and use
Lemma~\ref{lemma:committeeMatching} to test whether there exists a
completion $\T$ such that $s(\T,c) \geq s$, and for every $d \in D$ we
have $s(\T,d) \leq s$ if $c \tie d$ or $s(\T,d) < s$ otherwise. We conclude that:

\begin{theorem}\label{thm:ptwk-fixed-possible-veto-plurality}
For every fixed $k \geq 1$,
  $\decision{\minp}{{<}k}$ is solvable in polynomial time under the plurality and veto rules.
\end{theorem}

The polynomial  degree in Theorem~\ref{thm:ptwk-fixed-possible-veto-plurality} depends on
$k$. The following result shows that this is unavoidable, at least for the plurality rule, under conventional assumptions in parameterized complexity. 

\def\thmpluralitywhard{ Under the plurality rule, $\decision{\minp}{<}$ is $\mathrm{W}[2]$-hard for the parameter $k$. }
\begin{theorem}
\label{thm:plurality-w2-hard}
\thmpluralitywhard
\end{theorem}
\begin{proof}
We show an FPT reduction from the \e{dominating set} problem, which is the following: Given an undirected graph $G = (U,E)$ and an integer $k$, is there a set $D \subseteq U$ of size $k$ such that every vertex is either in $D$ or adjacent to some vertex in $D$? This problem is known to be W[2]-hard for the parameter $k$~\cite{DBLP:series/mcs/DowneyF99}.

Given $G = (U,E)$, we construct an instance of
$\decision{\minp}{<}$ under plurality where the candidate set is
$C = U \cup \set{c^*}$, the tiebreaker is $\tie = O(\set{c^*}, U)$,
and the voting profile is $\P = \set{P_u}_{u \in U}$ where $P_u$ is
defined as follows. Let $N(u)$ be
the set of neighbours of $u \in U$ and $N[u] = N(u) \cup
\set{u}$. We define
$P_u \eqdef P(N[u], U \setminus N[u], \set{c^*})$.  Hence, the voter
with preferences $P_u$ can vote only for vertices that dominate $u$. To complete, we show that the graph has a dominating set of size $k$
if and only if $\minrnk {c^*}\P\tie < k+2$.

Suppose there is a dominating set $D$ of size $k$, consider the profile $\T = \set{T_u}_{u \in U}$ where for every $u \in U$,
$$ T_u \eqdef O(N[u] \cap D, N[u] \setminus D, U \setminus N[u], \set{c^*})\,.$$
In this completion, for each $u \notin D$ we get $s(\T, u) = 0$. These
are $n-k$ candidates that $c^*$ defeats, therefore $\rnk {c^*}\T\tie < k+2$ (at most $k$ candidates defeat $c^*$) and $\minrnk {c^*}\P\tie < k+2$. Conversely, if $\minrnk {c^*}\P\tie < k+2$ then in some completion $\T$ it defeats at least $n-k$
candidates, and these candidates have a score 0 in $\T$. Let $D$ be the
set of candidates that $c^*$ does not defeat in $\T$, all voters voted
for candidates in $D$ and $|D| \leq k$. A voter $P_u$ can only vote
for vertices which dominate $u$, hence $D$ is a dominating set of
size at most $k$.
\end{proof}

\paragraph{Beyond Plurality and Veto}

The Classification Theorem (Theorem~\ref{thm:classification}) states
that \PW is intractable for every pure scoring rule $r$ other than
plurality or veto. While this hardness easily generalizes to $\decision{\minp_r}{{<}k}$ for $k=2$, it is not at all clear how to generalize it to any
$k>2$. In particular, we cannot see how to reduce \PW to $\decision{\minp_r}{{<}k}$ while assuming only the purity of the rule.  We can, however, show such a reduction under a stronger notion of purity, as long as the scores are bounded by a polynomial
in the number $m$ of candidates. In this case, we say that the rule
has \e{polynomial scores}.  
Note that all of the specific rules mentioned so far (i.e.,
$t$-approval, $t$-veto, Borda and so on) have polynomial scores;
an example of
a rule that does \e{not} have polynomial scores is $r(m,j)=2^{m-j}$.
Also note that this assumption is made in addition to
our usual assumption that the scores can be computed in polynomial
time. 

A rule $r$ is \e{strongly pure} if the score sequence for $m+1$
candidates is obtained from the score sequence for $m$ candidates by
inserting a new score to \e{either the beginning or the end of the
  sequence}. More formally, $r=\set{ \Vec{s}_m }_{m \in \natural^+}$ is strongly pure
if for all $m \geq 1$, either
$\Vec{s}_{m+1} = \Vec{s}_{m+1}(1) \circ \Vec{s}_m$ or
$\Vec{s}_{m+1} = \Vec{s}_m \circ \Vec{s}_{m+1}(m+1)$. Note that
$t$-approval, $t$-veto and Borda are all strongly pure.

{\begin{figure}
   \centering
\begin{tabular}{|c|c|c|c|c|c|c|c|c|}\hline
  Voter & 1 & 2 & $\cdots$ & $k-1$ & $k$ & $k+1$ & $\cdots$ & $k+m-1$ \\
  \hline
  $M_{1,1}$ & \cellcolor{lipicsLightGray} $d_1$ & $d_2$ & $\dots$ & $d_{k-1}$ & \cellcolor{lipicsYellow} $c_1$ & $c_2$ & $\cdots$ & $c_m$ \\
  \hline
  $M_{1,2}$ & \cellcolor{lipicsLightGray} $d_1$ & $d_2$ & $\dots$ & $d_{k-1}$ & $c_2$ & $c_3$ & $\cdots$ & \cellcolor{lipicsYellow} $c_1$ \\
  \hline
  $\vdots$ & & & & & & & & \\
  \hline
  $M_{1,m}$ & \cellcolor{lipicsLightGray} $d_1$ & $d_2$ & $\dots$ & $d_{k-1}$ & $c_m$ & \cellcolor{lipicsYellow} $c_1$ & $\cdots$ & $c_{m-1}$ \\
  \hline
  $M_{2,1}$ & $d_2$ & $d_3$ & $\dots$ & \cellcolor{lipicsLightGray} $d_1$ & \cellcolor{lipicsYellow} $c_1$ & $c_2$ & $\cdots$ & $c_m$ \\
  \hline
  $M_{2,2}$ & $d_2$ & $d_3$ & $\dots$ & \cellcolor{lipicsLightGray} $d_1$ & $c_2$ & $c_3$ & $\cdots$ & \cellcolor{lipicsYellow} $c_1$ \\
  \hline
  $\vdots$ & & & & & & & & \\
  \hline
  $M_{2,m}$ & $d_2$ & $d_3$ & $\dots$ & \cellcolor{lipicsLightGray} $d_1$ & $c_m$ & \cellcolor{lipicsYellow} $c_1$ & $\cdots$ & $c_{m-1}$ \\
  \hline
  $\vdots$ & & & & & & & & \\
   \hline
  $M_{k-1,1}$ & $d_{k-1}$ & \cellcolor{lipicsLightGray} $d_1$ & $\dots$ & $d_{k-2}$ & \cellcolor{lipicsYellow} $c_1$ & $c_2$ & $\cdots$ & $c_m$ \\
  \hline
  $M_{k-1,2}$ & $d_{k-1}$ & \cellcolor{lipicsLightGray} $d_1$ & $\dots$ & $d_{k-2}$ & $c_2$ & $c_3$ & $\cdots$ & \cellcolor{lipicsYellow} $c_1$ \\
  \hline
  $\vdots$ & & & & & & & & \\
  \hline
  $M_{k-1,m}$ & $d_{k-1}$ & \cellcolor{lipicsLightGray} $d_1$ & $\dots$ & $d_{k-2}$ & $c_m$ & \cellcolor{lipicsYellow} $c_1$ & $\cdots$ & $c_{m-1}$ \\
  \hline
\end{tabular}
\caption{\label{fig:Mij} The voters $M_{i,j}$ used in the proof of Theorem~\ref{thm:PolyRulePosMemk}.}
\end{figure}}

\def\thmPolyRulePosMemk{ Suppose that a positional scoring rule is strongly pure, has polynomial scores, and is neither plurality nor veto. Then $\decision{\minp}{{<}k}$ is NP-complete for all fixed $k \geq 2$. }
\begin{theorem}
\label{thm:PolyRulePosMemk}
\thmPolyRulePosMemk
\end{theorem}

\begin{proof}
Let $r=\set{\Vec{s}_m}_{m>1}$ be a positional scoring rule that satisfies the conditions of the theorem, and let $k \geq 1$. We show a reduction from \PW under $r$ to $\decision{\minp_r}{{<\,}{k+1}}$. The idea is to add $k-1$ new candidates and modify the voters so that the new candidates are always the top $k-1$ candidates, and the score of each of the original candidates is increased by the same amount.

  Consider the input
  $\P = (P_1, \dots, P_n)$ and $c$ for $\PW$ over a set $C$ of $m$
  candidates.  Let $m' = m+k-1$. Since $r$ is strongly pure, there is
  an index $t \leq k-1$ such that
$$ \Vec{s}_{m'} = (\Vec{s}_{m'}(1), \dots, \Vec{s}_{m'}(t)) \circ \Vec{s}_m \circ (\Vec{s}_{m'}(t+m+1), \dots, \Vec{s}_{m'}(m'))\,. $$
That is, $\Vec{s}_{m'}$ is obtained from $\Vec{s}_m$ by inserting $t$
values at the top coordinates and $k-1-t$ values at the bottom
coordinates. We define $C'$, $\P'$ and $\tieprime$ as follows. 

The candidate set is $C' = C \cup D_1 \cup D_2$ where $D_1 = \set{d_1,
  \dots, d_t}$ and $D_2 = \set{d_{t+1}, \dots, d_{k-1}}$.
Denote $D = D_1 \cup D_2$. The tiebreaker is $\tieprime = O(D, \set{c}, C \setminus \set{c})$. The profile $\P'$ is the concatenation $\Q\circ\M$ of two voting profiles. The first is $\Q=(Q_1, \dots, Q_n)$, where $Q_i$ is the same as $P_i$, except that the candidates of $D_1$ are placed at the top positions and the candidates of $D_2$ are placed at the bottom positions. Formally, $Q_i \eqdef P_i \cup P(D_1, C, D_2)$. The second, $\M$, consists of $n \cdot \Vec{s}_{m'}(1)$ copies of the profile $\set{M_{i,j}}_{i = 1, \dots, k-1 \,,\, j = 1, \dots, m}$
where $M_{i,j}$ is $M_i(D) \circ M_j(C)$ for the circular votes
$M_i(D)$ and $M_j(C)$ as defined in the proof of
Theorem~\ref{thm:minAllRules}. (See Figure~\ref{fig:Mij}.)  Note that by the conditions of the theorem, $n \cdot \Vec{s}_{m'}(1)$ is polynomial in $n, m$.

We show that the candidates of $D$ always defeat all other
candidates. For every $d \in D$, the score of $d$ in $\M$ is $s(\M, d) = n \cdot \Vec{s}_{m'}(1) \cdot m \sum_{i=1}^{k-1} \Vec{s}_{m'}(i)$, and for every $c' \in C$ the score in $\M$ is
\begin{align*}
    s(\M, c') &= n \cdot \Vec{s}_{m'}(1) 
    \cdot (k-1) \sum_{i=k}^{m'} \Vec{s}_{m'}(i) \leq n \cdot \Vec{s}_{m'}(1) 
    \cdot \rpar{m \sum_{i=1}^{k-1} \Vec{s}_{m'}(i)-1} \\
    &= s(\M, d) - n \cdot \Vec{s}_{m'}(1)
\end{align*}
where the inequality is due to the assumption that
$\Vec{s}_{m'}(1) > \Vec{s}_{m'}(m')$. Let $\T'$ be a completion of $\P'$, the total score of $c'$ is
\begin{align*}
    s(\T', c') &\leq n \cdot \Vec{s}_{m'}(1) + s(\M, c') \leq n \cdot \Vec{s}_{m'}(1) + s(\M, d) - n \cdot \Vec{s}_{m'}(1)  \leq s(\T', d)\,.
\end{align*}
Since the candidates of $D$ are the first candidates in $\tieprime$,
they always defeat the candidates of $C$. Next, we show that $c$ is a possible winner for $\P$ if and only if $\minrnk {c}{\P'}{\tieprime} = k$. Since the candidates of $D$ are always the first $k-1$ candidates, this is equivalent to saying that $\minrnk {c}{\P'}{\tieprime} < k+1$.

Let $\T = (T_1, \dots T_n)$ be a completion of $\P$ where $c$ is a winner. Consider the
completion $\T' = (T'_1, \dots T'_n) \circ \M$ of $\P'$ where
$T'_i = O(D_1) \circ T_i \circ O(D_2)$. For every $c' \in C$, we know that
$s(\T', d) \geq s(\T', c')$ for every $d \in D$, and
from the property of $\Vec{s}_{m'}$ we get that
\[ s(\T', c') = s(\T, c') + n \cdot \Vec{s}_{m'}(1) \cdot (k-1)
\sum_{i=k}^{m'} \Vec{s}_{m'}(i)\,. \]
From the choice of $\tieprime$, $c$ defeats all candidates of $C \setminus \set{c}$ in $\T'$, hence $\rnk c{\T'}{\tieprime} = k$. Conversely, let $\T' = (T'_1, \dots T'_n) \circ \M$ be a completion of $\P'$ where $\rnk c{\T'}{\tieprime} = k$, define a completion $\T$ of $\P$ by removing $D$ from all orders in $(T'_1, \dots T'_n)$. For every $c' \in C$ we have
\[ s(\T, c') = s(\T', c') - n \cdot \Vec{s}_{m'}(1) \cdot (k-1)
\sum_{i=k}^{m'} \Vec{s}_{m'}(i) \]
hence $c$ is a winner in $\T$.
\end{proof}

\subsubsection{Complexity of $\decision{\maxp}{{>}k}$}

The following theorem states that $\decision{\maxp}{{>}k}$ is tractable for
every fixed $k$ and every positional scoring
rule (pure or not) with polynomial scores. 

\begin{theorem} \label{thm:PolyRuleNecMemk} For all fixed $k \geq 1$ and positional scoring rules $r$ with polynomial scores, $\decision{\maxp_r}{{>}k}$ is solvable in polynomial time.
\end{theorem}

Next, we prove
Theorem~\ref{thm:PolyRuleNecMemk}. To determine whether $\maxrnk c\P\tie > k$, we search for $k$ candidates that defeat $c$ in some completion $\T$, since $\rnk c\T\tie > k$ if and only if at least $k$ candidates defeat $c$ in $\T$. For that,
we consider each subset
$\set{c_1, \dots, c_k} \subseteq C \setminus \set{c}$ and determine
whether these $k+1$ candidates can get a combination of scores where $c_1, \dots, c_k$ all defeat $c$.

More formally, let $C$ be a set of candidates and $r$ a positional
scoring rule. For a partial profile $\P = (P_1, \dots, P_n)$ and a
sequence $S=(c_1,\dots,c_q)$ of candidates from $C$, we denote by
$\ps(\P,S)$ the set of all possible scores that the candidates in $S$
can obtain jointly in a completion:
$ \ps(\P,S) \eqdef \set{(s(\T,c_1), \dots, s(\T,c_q)) \mid \T \text{
    completes } \P} $.  Note that
$\ps(\P,S) \subseteq\set{0,\dots,n \cdot \Vec{s}_m(1)}^q$.  When $\P$
consists of a single voter $P$, we write $\ps(P,S)$ instead of
$\ps(\P,S)$. To show that $\maxrnk c\P\tie > k$ we need to find a sequence $S=(c_1,\dots,c_q)$ of distinct candidates where $q=k+1$ and $c_q=c$, and a sequence $(s_1,\dots,s_q)\in\ps(\P,S)$ such that the following holds: for $i=1,\dots k$ we have $s_i \geq s_q$ if $c_i \tie c_q$ and $s_i > s_q$ if $c_q \tie c_i$. The following two lemmas show that  if such a sequence exists, then we can find it in polynomial time.

\begin{lemma}
  \label{lemma:scheduling}
  Let $q$ be a fixed natural number and $r$ a positional scoring rule.
  Whether $(s_1, \dots, s_q) \in \ps(P,S)$ can be determined in
  polynomial time, given a partial order $P$ over a set of candidates,
  a sequence $S$ of $q$ candidates, and scores
  $s_1,\dots, s_q$.
\end{lemma}
\begin{proof}
  We use a reduction to a scheduling problem where tasks have
  \e{execution times}, \e{release times}, \e{deadlines}, and
  \e{precedence constraints} (i.e., task $x$ should be completed
  before starting task $y$).  This scheduling problem can be solved in
  polynomial time~\cite{DBLP:journals/siamcomp/GareyJST81}.  In the
  reduction, each candidate $c$ is a task with a unit execution time.
  For every $c_i$ in $S$, the release time is $\min \set{j \in [n] : r(m,j) = s_i}$,
  and the deadline is
  $1+\max \set{j \in [n] : r(m,j) = s_i}$. For the rest of the candidates, the release time is 1 and the deadline is $m+1$. The precedence constraints are
  $P$. It holds that $(s_1, \dots, s_q) \in \ps(P,S)$ if and only if the
  tasks can be scheduled according to all the requirements.
\end{proof}

From Lemma~\ref{lemma:scheduling} we can conclude that when $q$ is
fixed and $r$ has polynomial scores, we can construct $\ps(\P,S)$ in
polynomial time, via simple dynamic programming.

\def\lemmadynamic{ Let $q$ be a fixed natural number and $r$ a positional scoring rule with polynomial scores. The set $\ps(\P,S)$ can be computed in polynomial time, given a partial profile $\P$ and a sequence $S$ of $q$ candidates. }
\begin{lemma}
\label{lemma:dynamic}
\lemmadynamic
\end{lemma}
\begin{proof}
First, for every $i \in [n]$, construct $\ps (P_i, S)$ by applying Lemma~\ref{lemma:scheduling}. Then, given  $\ps ((P_1, \dots, P_i), S)$, observe that
$$ \ps ((P_1, \dots, P_{i+1}), S) = \set{ \Vec{u} + \Vec{w} : \Vec{u} \in \ps ((P_1, \dots, P_i), S), \Vec{w} \in \ps (P_{i+1}, S)}$$
where $\Vec{u} + \Vec{w}$ is a point-wise sum of the two vectors $(\Vec{u} + \Vec{w})(j) = \Vec{u}(j)+ \Vec{w}(j)$. For the complexity, recall that the size of $\ps(\P,S)$ is $O((n \cdot \Vec{s}_m(1))^q)$. Since $r$ has polynomial scores, we get that the size of $\ps ((P_1, \dots, P_i), S)$ is polynomial for every $i \in [n]$. Hence, $\ps(\P,S)$ can be constructed in polynomial time. 
\end{proof}

From Lemma~\ref{lemma:dynamic} we conclude
Theorem~\ref{thm:PolyRuleNecMemk}. Note that the polynomial degree depends on $k$. This is unavoidable under conventional assumptions in parameterized complexity---we can modify the proof of Theorem~\ref{thm:maxAllRules} to get an FPT reduction from the regular clique problem, which is W[1]-hard~\cite{DBLP:conf/cats/MathiesonS08}, to $\decision{\maxp}{>}$. Therefore:

\def\thmmaxWHard{ For every positional scoring rule, $\decision{\maxp}{>}$ is $\mathrm{W}[1]$-hard for the parameter $k$. }
\begin{theorem}
\label{thm:max-w1-hard}
\thmmaxWHard
\end{theorem}
\begin{proof}
We describe an FTP reduction, in two steps, from the regular clique problem: Given a regular graph $G = (U, E)$ and an integer $k$, is there a subset $U' \subseteq U$ of size at least $k$ such that every two vertices in $U'$ are joined by an edge in $E$? This problem is $\mathrm{W}[1]$-hard for the parameter $k$~\cite{DBLP:conf/cats/MathiesonS08}.

First, by the standard reduction from clique to independent set, deciding whether a regular graph contains an independent set of size $k$ is $\mathrm{W}[1]$-hard for the parameter $k$. Then, the proof of Theorem~\ref{thm:maxAllRules} provides an FTP reduction from independent set on regular graphs to $\decision{\maxp}{>}$.
\end{proof}

\subsubsection{Complexity of $\decision{\minp}{{<}\bark}$}

Recall that $\bark \eqdef m-k+1$. We now show that the problem of
$\decision{\minp}{{<}\bark}$ is tractable for every positional
scoring rule with polynomial scores. We find it surprising because
$\decision{\minp}{{<}2}$ is NP-complete for every pure positional
scoring rule other than plurality and veto, by a reduction from \PW.

Given a positional scoring rule $r$ and functions $a, b : \natural_+ \rightarrow \natural_+$, we define the
\e{$(a,b)$-reversed} scoring rule, denoted $r^{a, b}$, to be the one
given by $r^{a, b}(m,i) = a(m) - b(m) \cdot r(m,m+1-i)$. For example, the
$(1, 1)$-reversed rule of plurality is veto, and more generally,
the $(1, 1)$-reversed rule of $t$-approval is $t$-veto. Also, the $(m,
1)$-reversed rule of Borda is Borda itself.
In the following lemma, we use a generalized notation for our decision
problems where, instead of fixed $k$ or $\bark$, we use a fixed
function $f$ that is is applied to the number $m$ of candidates to
produce a number $f(m)$.

\def\lemmareversing{
Let $r$ be a positional scoring rule, let $f, a, b : \natural_+ \rightarrow \natural_+$, and let $\bar{f}(m) = m+1-f(m)$. There exists a reduction from $\decision{\minp_r}{{<} f}$ to $\decision{\maxp_{r^{a, b}}}{{>}\bar{ f}}$, and from $\decision{\maxp_r}{{>} f}$ to $\decision{\minp_{r^{a, b}}}{{<} \bar{f}}$. }
\begin{lemma}
\label{lemma:reversing}
\lemmareversing
\end{lemma}
\begin{proof}
For a partial order $P$, the \e{reversed order} is defined by 
$P^R \eqdef \set{x \succ y : (y \succ x) \in P}$. Note that $T$ extends $P$ if and only if $T^R$ extends $P^R$. Given $(C, \P, \tie)$ as input under $r$ with $\P = (P_1, ..., P_n)$, consider $(C, \P', \tie^R)$ under $r^R$ where $\P' = (P_1^R, ..., P_n^R)$. 

Let $\T = (T_1, ..., T_n)$ be a completion of $\P$, observe the completion $\T' = (T_1^R, ..., T_n^R)$ of $\P'$. For every candidate $c$ and voter $v_i$ we get $s(T_i^R, c, r^{a,b}) = a(m) - b(m) \cdot s(T_i, c, r)$ so overall $s(\T', c, r^{a,b}) = n \cdot a(m) - b(m) \cdot s(\T, c, r)$. Since the tie-breaking order is also reversed, the rank is $\rnk c{\T'}{\tie^R} = m+1-\rnk c\T\tie$. In same way, if $\T'$ is a completion of $\P'$ then by reversing the orders we get a completion $\T$ of $\P$ such that $\rnk c\T\tie = m+1-\rnk c{\T'}{\tie^R}$ for every $c \in C$. We can deduce that $\minrnk c\P\tie < f(m)$ if and only if $\maxrnk c{\P'}{\tie^R} > f'(m)$, and $\maxrnk c\P\tie > f(m)$ if and only if $\minrnk c{\P'}{\tie^R} < f'(m)$. From the above points we conclude the parts of the lemma, respectively.
\end{proof}

Using Lemma~\ref{lemma:reversing} and
Theorem~\ref{thm:PolyRuleNecMemk}, we can show that:

\def\thmpolyruleboundedreversedmin{ $\decision{\minp}{{<}\bark}$ is solvable in polynomial time for every fixed $k \geq 1$ and positional scoring rule $r$ with polynomial scores. }
\begin{theorem}
\label{thm:poly-rule-bounded-reversed-min}
\thmpolyruleboundedreversedmin
\end{theorem}
\begin{proof}
Let $r$ be a positional scoring rule with polynomial scores, and denote $r$ by $\set{\Vec{s}_m}_{m>1}$. Define the functions $a(m) = \Vec{s}_m(1)$, $b(m) = 1$, and observe $r^{a, b}$. For any $m>1$, the vector for $r^{a,b}$ is
$$ (\Vec{s}_m(1) - \Vec{s}_m(m), \Vec{s}_m(1) - \Vec{s}_m(m-1), \dots, \Vec{s}_m(1) - \Vec{s}_m(2), \Vec{s}_m(1) - \Vec{s}_m(1)) $$
therefore $r^{a,b}$ is also with polynomial scores. For any fixed $k$,
$\decision{\maxp_{r^{a,b}}}{{>}k}$ is solvable in polynomial
time by Theorem~\ref{thm:PolyRuleNecMemk}. Then, by
Lemma~\ref{lemma:reversing}, $\decision{\minp_r}{{<}\bark}$ is
solvable in polynomial time.
\end{proof}

\subsubsection{Complexity of $\decision{\maxp}{{>}\bark}$}

First, for plurality and veto, by Theorem~\ref{thm:ptwk-fixed-possible-veto-plurality} and Lemma~\ref{lemma:reversing}, we can deduce the following:
\begin{corollary}
\label{cor:max-fixed-reversed-veto-plurality}
For every fixed $k \geq 1$, $\decision{\maxp}{{>}\bark}$ is
solvable in polynomial time under the plurality and veto rules.
\end{corollary}

A positional scoring rule $r$ is \e{$p$-valued}, where $p$ is a
positive integer greater than $1$, if there exists a positive integer
$m_0$ such that for all $m \geq m_0$, the scoring vector $\Vec{s}_m$
of $r$ contains exactly $p$ distinct values. A rule is \e{bounded} if it
is $p$-valued for some $p > 1$. Note that for a pure bounded rule
there exists some constant $t$ such that for every $m$, the values in
$\Vec{s}_m$ are at most $t$, since for all $m > m_0$ the vector
$\Vec{s}_m$ cannot contain values that do not appear in
$\Vec{s}_{m_0}$. Combining Theorem~\ref{thm:PolyRulePosMemk} and Lemma~\ref{lemma:reversing},
we get the following:

\def\thmmaxboundedreversedstrongly{ Suppose that a positional scoring rule $r$ is bounded, strongly pure,
and is neither plurality nor veto. Then $\decision{\maxp_r}{{>}\bark}$ is NP-complete for all fixed
$k\geq 2$. }
\begin{theorem}
\label{thm:max-bounded-reversed-strongly}
\thmmaxboundedreversedstrongly
\end{theorem}
\begin{proof}
Let $r$ be a positional scoring rule that satisfies the conditions of the theorem, and let us denote $r$ by $\set{\Vec{s}_m}_{m>1}$. Since $r$ is bounded and strongly pure, there exists some constant $t$ such that for every $m$, the values in $\Vec{s}_m$ are at most $t$. Observe the scoring rule $r' = r^{t,1}$. For every $m \geq 1$, if $\Vec{s}_{m+1} = \Vec{s}_{m+1}(1) \circ \Vec{s}_m$ then the vector of $r'$ for $m+1$ candidates is
 \begin{align*}
     & (t-\Vec{s}_{m+1}(m+1), t-\Vec{s}_{m+1}(m), \dots, t-\Vec{s}_{m+1}(1)) \\
     &= (t-\Vec{s}_{m}(m), t-\Vec{s}_{m}(m-1), \dots, t-\Vec{s}_{m}(1)) \circ (t-\Vec{s}_{m+1}(1))
 \end{align*}
Otherwise, $\Vec{s}_{m+1} = \Vec{s}_m \circ \Vec{s}_{m+1}(m+1)$, and the vector of $r'$ for $m+1$ candidates is
 \begin{align*}
     & (t-\Vec{s}_{m+1}(m+1), t-\Vec{s}_{m+1}(m), \dots, t-\Vec{s}_{m+1}(1)) \\
     & = (t-\Vec{s}_{m+1}(n+1)) \circ (t-\Vec{s}_{m}(m), t-\Vec{s}_{m}(m-1), \dots, t-\Vec{s}_{m}(1))
\end{align*}
Therefore $r'$ is strongly pure, has polynomial scores (the scores are
bounded by $t$), and is neither plurality nor veto (because $r$ is
neither plurality nor veto). By Theorem~\ref{thm:PolyRulePosMemk},
$\decision{\minp_{r'}}{{<}k}$ is NP-complete. Since $r = (r')^{t,1}$, by Lemma~\ref{lemma:reversing} we
deduce that $\decision{\maxp_r}{{>}\bark}$ is NP-complete.
\end{proof}
\section{Additional Voting Rules}\label{sec:other}

In this section, we consider other, non-positional voting rules. In each rule, we recall the definition of the score of a candidate that is used for winner determination (i.e., top-score candidates). Once we have the score, we automatically get the rank of a candidate, namely $\rnk c\T\tie$, and the minimal and maximal ranks, namely $\minrnk c\P\tie, \maxrnk c\P\tie$, respectively, in the same way as the positional scoring rules. Our results are summarized in Table~\ref{tab:complexityOther}.

{\begin{table}[t]
  \def\arraystretch{1.2}
  \caption{\label{tab:complexityOther}  Results for non-positional voting rules.}
  \centering
  \scalebox{1}{
    \begin{tabular}{c | c c | c c | c c}
  \hline
  \textbf{Problem} & \multicolumn{2}{c|}{Copeland} & \multicolumn{2}{c|}{Bucklin} & \multicolumn{2}{c}{Maximin} \\
  \hline\hline
  \PW & NP-c & \cite{DBLP:journals/jair/XiaC11} & NP-c & \cite{DBLP:journals/jair/XiaC11} & NP-c & \cite{DBLP:journals/jair/XiaC11} \\
  \hline
  \NW & coNP-c & \cite{DBLP:journals/jair/XiaC11} & P & \cite{DBLP:journals/jair/XiaC11} & P & \cite{DBLP:journals/jair/XiaC11} \\
  \hline
  $\decision{\minp}{<}$ & NP-c & [Thm.~\ref{thm:copelandMinMax}] & NP-c & [Thm.~\ref{thm:bucklinMinMax}] & NP-c & [Thm.~\ref{thm:maximinMinBounded}] \\
 \hline
  $\decision{\maxp}{>}$ & NP-c & [Thm.~\ref{thm:copelandMinMax}] & NP-c & [Thm.~\ref{thm:bucklinMinMax}] & NP-c & [Thm.~\ref{thm:maximinMaxBounded}] \\
  \hline
  $\decision{\minp}{<k}$ & NP-c & [Thm.~\ref{thm:copelandMinMax}] & NP-c & [Thm.~\ref{thm:bucklinMinBounded}] & NP-c & [Thm.~\ref{thm:maximinMinBounded}] \\
  \hline
  $\decision{\maxp}{>k}$ & NP-c & [Thm.~\ref{thm:copelandMinMax}] & P & [Thm.~\ref{thm:bucklinMaxBounded}] & NP-c & [Thm.~\ref{thm:maximinMaxBounded}] \\
  \hline
  \end{tabular}}
\end{table}}

\subsection{Copeland}
We say that a candidate $c$ \e{defeats} $c'$ in a pairwise election if the majority of the votes rank $c$ ahead of $c'$. In the Copeland rule, the score of $c$ is the number of candidates $c' \neq c$ that $c$ defeats in a pairwise election. A winner is a candidate with a maximal score. It is known that $\PW$ is NP-complete and $\NW$ is coNP-complete with respect to Copeland~\cite{DBLP:journals/jair/XiaC11}. We use reductions from \PW and \NW under Copeland to obtain hardness of computing the minimal and maximal ranks, respectively.

\def\thmcopelandMinMax{ For the Copeland rule, $\decision{\minp}{{<}k}$ is NP-complete for all fixed $k \geq 2$, and  $\decision{\maxp}{{>}k}$ is NP-complete for all fixed $k \geq 1$. }
\begin{theorem}
\label{thm:copelandMinMax}
\thmcopelandMinMax
\end{theorem}
\begin{proof}
Let $k \geq 1$. For minimal rank, we use a reduction from $\PW$ under Copeland. Consider the input $\P = (P_1, \dots, P_n)$ and $c$ for $\PW$ over a set $C$ of $m$ candidates. We define $C', \P'$ and $\tie$ as follows. The candidates are $C' = C \cup D$ where $D = \set{d_1, \dots, d_{k-1}}$. The tie breaker is $\tie = O(D, \set{c}, C \setminus \set{c})$. The profile is $\P' = (P_1', \dots, P_n')$ where $P_i'$ is the same as $P_i$, except that the candidates of $D$ are placed at the top positions. Formally, $P_i' \eqdef P_i \cup P(D, C)$. 

Observe that the candidates of $D$ defeat all candidates of $C$ in pairwise election, for every completion $\T'$ of $\P'$. Furthermore, for every completion $\T$ of $\P$ we can easily construct a completion $\T'$ of $\P'$ that satisfies $s(\T, c') = s(\T', c')$ for all $c' \in C$, and vice versa. We conclude that $c$ is a possible winner for $\P$ if and only if $\minrnk {c}{\P'}{\tieprime} < k+1$. For the maximal rank, we use a reduction from $\NW$ under Copeland, which is the same as for the minimal rank. If $c$ is a necessary winner, then $\maxrnk {c}{\P'}{\tieprime} < k+1$; otherwise, $\maxrnk {c}{\P'}{\tieprime} > k$. \end{proof}

\subsection{Bucklin}
Under the Bucklin rule, the score of a candidate $c$ is the smallest number $t$ such that more than half of the voters rank $c$ among the top $t$ candidates. A winner is a candidate with a \e{minimal} Bucklin score. Since we prefer the minimal score rather than the maximal score, we need to modify the definition of the rank: Let $R_\T'$ be the linear order on $C$ that sorts the candidates by their scores in \e{increasing} order and then by $\tie$. The rank of $c$ is the position of $c$ in $R_\T'$, which we denote again by $\rnk c \T \tie$. It is known that $\PW$ is NP-complete and $\NW$ is in polynomial time with respect to Bucklin \cite{DBLP:journals/jair/XiaC11}.  We show that computing the minimal and maximal ranks is hard for the Bucklin rule.

\def\thmbucklinMinMax{ For the Bucklin rule, both $\decision{\minp}{{<}}$ and $\decision{\maxp}{{>}}$ are NP-complete. }
\begin{theorem}
\label{thm:bucklinMinMax}
\thmbucklinMinMax
\end{theorem}
\begin{proof}
First, $\decision{\minp}{{<}}$ is NP-complete for Bucklin by a straightforward reduction from $\PW$ (as $\PW$ is the special case of $\decision{\minp}{{<}2}$). For $\decision{\maxp}{{>}}$. We show a reduction from the independent-set problem in 3-regular graphs, as defined in the proof of Theorem~\ref{thm:maxAllRules}. Let $G = (U, E)$ be a 3-regular graph with $U = \set{u_1, \dots, u_n}$. We construct an instance $(C, \P, \tie)$ under Bucklin. The candidates set is $C = U \cup \set{c^*, d} \cup F$ where $F = \set{f_1, \dots, f_{n-1}}$ and the tiebreaker is $\tie = O(F, U, \set{c^*}, \set{d})$. The voting profile $\P = \P^1 \circ \T^2$ is the concatenation  of two parts described next. 

The first part, $\P^1 = \set{P^1_e}_{e \in E}$, contains a voter for every edge $e$. For each edge $e = \set{u,w} \in E$, define $P^1_e = P(F, \set{c^*}, e, U \setminus e, \set{d})$. Then, in three arbitrary voters in $\P^1$, switch between $c^*$ and $d$ (i.e., the profile becomes $P^1_e = P(F, \set{d}, e, U \setminus e, \set{c^*})$). The second part is $\T^2 = (T^2_1, \dots, T^2_{|E|-4})$ where every voter is $T^2_i = O(U, \set{c^*}, F, \set{d})$. Overall, there are $2n+1$ candidates and $2|E|-4$ voters.

Next, we state some observations. Let $\T = \T^1 \circ \T^2$ be a completion of $\P$ where $\T^1 = \set{T^1_e}_{e \in E}$.
\begin{itemize}
    \item In $\T^1$, there are $|E|-3$ voters who place $c^*$ at the $n$th position, and 3 voters who place $c^*$ at the $(2n+1)$th position. In $\T^2$, there are $|E|-4$ voters who place $c^*$ at the $(n+1)$th position. Overall, less than half of the voters rank $c^*$ among the top $n$, and more than half of the voters rank $c^*$ among the top $n+1$. Hence $s(\T, c^*) = n+1$.
    \item For every $f \in F$, there are $|E|$ voters in $\T^1$ who rank $f$ among the top $n-1$ positions, hence $s(\T, f) \leq n-1$.
    \item There are 3 voters who place $d$ at the $n$th positions, and all other $2|E|-7$ voters place $d$ at the $(2n+1)$th position, hence $s(\T, d) = 2n+1$. \bennydone{Why is it an equality?}
    \item For every $v \in V$, there are $|E|-4$ voters in $\T^2$ who rank $v$ among the top $n$. If all three edges of $v$ place $v$ at the $(n+1)$th position, we get $|E|-1$ voters which rank $v$ among the top $n+1$, and $s(\T, v) = n+1$. Otherwise, at most $|E|-2$ voters rank $v$ among the top $n+1$, which implies $s(\T, v) \geq n+2$.
\end{itemize}

Overall, the candidates of $F$ defeat all other candidates, and $d$ is defeated by all other candidates. Every $v \in V$ defeats $c^*$ if and only if all edges of $v$ rank it at the $(n+1)$th position. As in the proof of Theorem~\ref{thm:maxAllRules}, we denote the maximal size of an independent set in $G$ by $\beta(G)$. We now show that for every $k$ it holds that $\beta(G) \geq k$ if and only if $\maxrnk {c^*}\P\tie \geq k + |F| + 1$, which implies NP-completeness for $\decision{\maxp}{>}$.

Assume that $\beta(G) \geq k$, let $B$ be an independent set of size at least $k$ in $G$. Consider a completion  $\T = \set{T^1_e}_{e \in E} \circ \T^2$ of $\P$ as follows. For every $e = \set{u, w} \in E$, recall that only the positions of $u,w$ are not determined in the voter $P^1_e$. If $u, w \notin B$ then complete $P^1_e$ arbitrarily.  If $u \in B$ then in $T^1_e$, $u$ is placed at the $(n+1)$th position and $w$ is placed at the $(n+2)$th position. Finally, if $w \in B$ then in $T^1_e$, $w$ is placed at the $(n+1)$th position and $u$ is placed at the $(n+2)$th position. Note that we cannot have $u, w \in B$ because $B$ is an independent set.

For every $u \in B$, all three edges of $u$ place $u$ at the $(n+1)$th position, hence as we said $s(\T, u) = n+1 = s(\T, c^*)$ and $u$ defeats $c^*$. These are at least $k$ candidates of $U$ which defeat $c^*$, and the candidates of $F$ always defeat $c^*$, therefore $\rnk {c^*}\T\tie \geq k + |F| + 1$ and $\maxrnk {c^*}\P\tie \geq k + |F| + 1$.

Next, assume that $\maxrnk {c^*}\P\tie \geq k + |F| + 1$, there exists a completion $\T = \set{T^1_e}_{e \in E} \circ \T^2$ of $\P$ where $\rnk {c^*}\T\tie \geq k + |F| + 1$. Let $B$ be the candidates of $U$ which defeat $c^*$ in $\T$. We know that the candidates of $F$ always defeat $c^*$ and $d$ never defeats $c^*$, hence $|B| \geq k$. For every $u \in B$ the score is $s(\T, u) = n+1$, otherwise $u$ does not defeat $c^*$.

Assume to the contrary that  $e = \set{u, w} \in E$ for some pair $u, w \in B$. In $T^1_e$, both $u$ and $w$ should be placed at the $(n+1)$th position, that is a contradiction. Hence $B$ is an independent set, which implies $\beta(G) \geq k$.
\end{proof}

For comparing the minimal and maximal ranks to some fixed rank $k$, we show that the complexity is the same as of \PW and \NW under Bucklin. Namely,  $\decision{\minp}{{<}k}$ is NP-complete and $\decision{\maxp}{{>}k}$ can be solved in polynomial time. The following theorem states the hardness of $\decision{\minp}{{<}k}$ under Maximin. The proof is by the same reduction as that of Theorem~\ref{thm:copelandMinMax}.

\def\thmbucklinMinBounded{ For the Bucklin rule, $\decision{\minp}{{<}k}$ is NP-complete for all fixed $k \geq 2$. }
\begin{theorem}
\label{thm:bucklinMinBounded}
\thmbucklinMinBounded
\end{theorem}
\begin{proof}
Let $k \geq 1$. We use a reduction from $\PW$ under Bucklin, which is the same reduction from Theorem~\ref{thm:copelandMinMax}. Recall that given input $\P$ and $c$ for $\PW$, we define $C', \P'$ and $\tie$. Let $\T'$ be a completion of $\P'$. For every $d \in D$, all voters rank $d$ among the top $k-1$ candidates, hence $s(\T, d) \leq k-1$. For every $c \in C$, the rank of $c$ in each voter is at least $k$, hence  $s(\T, c) \geq k$. Furthermore, for every completion $\T$ of $\P$ we can easily construct a completion $\T'$ of $\P'$ which satisfies $s(\T, c') = s(\T', c') - (k-1)$ for all $c' \in C$, and vice versa. Therefore $c$ is a possible winner for $\P$ if and only if $\minrnk {c}{\P'}{\tieprime} < k+1$.
\end{proof}

In contrast, we show that $\decision{\maxp}{{>}k}$ is solvable in polynomial time for all fixed $k \geq 1$ under the Bucklin rule. The proof follows a strategy similar to the proof of Theorem~\ref{thm:PolyRuleNecMemk}. For a partial profile $\P = (P_1, \dots, P_n)$, a sequence of candidates $S = (c_1, \dots, c_q) \in C^q$ and a sequence of positions $R = (r_1, \dots, r_q) \in [m]^q$, let $\ps(\P, S, R)$ be the set of vectors $(h_1, \dots, h_q) \in \set{0, \dots, n}^q$ for which there exists a completion $\T$ of $\P$ where for every $c_i$ in $S$, there are $h_i$ voters who rank $c_i$ among the top $r_i$ positions. Note that $\ps(\P, S, R) \subseteq \set{0, \dots, n}^q$. When $\P$ consists of a single voter $P$, we write $\pi(P, S, R)$ instead of $\pi(\P, S, R)$. The following lemma shows that we can construct $\ps(\P, S, R)$ in polynomial time for a fixed number $q$, using scheduling with precedence constraints and dynamic programming.

\def\lemmabucklinScheduling{ Let $q$ be a fixed number. Whether $(h_1, \dots, h_q) \in \ps(P,S, R)$ can be determined in polynomial time, given a partial order $P$ over a set of candidates, a sequence $S$ of $q$ candidates, a sequence $R$ of $q$ positions, and numbers $h_1,\dots, h_q$. }
\begin{lemma}
\label{lemma:bucklinScheduling}
\lemmabucklinScheduling  
\end{lemma}
\begin{proof}
Recall that $h_i$ represents the number of voters that place $c_i$ among the top $r_i$ positions in a completion. Since we have a single voter $P$, if $h_i \notin \set{0,1}$ for some $i \in [q]$ then $(h_1, \dots, h_q) \notin \ps(P,S, R)$. From now we assume that $h_i \in \set{0,1}$ for every $i \in [q]$. We use a reduction to a scheduling problem where tasks have \e{execution times}, \e{release times}, \e{deadlines}, and \e{precedence constraints}, as in the proof of Lemma~\ref{lemma:scheduling}. In the reduction, each candidate $c$ is a task with a unit execution time.

For every $c_i$ in $S$ with $h_i = 1$, we need to place $c_i$ among the top $r_i$ positions, hence the release time is $1$  and the deadline is $r_i$. For every $c_i$ in $S$ with $h_i = 0$, $c_i$ should not be among the top $r_i$ positions, hence the release time is $r_i+1$ and the deadline is $m+1$. For the rest of the candidates, the release time is 1 and the deadline is $m+1$. The precedence constraints are $P$. It holds that $(h_1, \dots, h_q) \in \ps(P,S, R)$ if and only if the tasks can be scheduled according to all the requirements.
\end{proof}

Next, we use Lemma~\ref{lemma:bucklinScheduling} to prove Lemma~\ref{lemma:bucklinDynamic}.

\def\lemmabucklinDynamic{ Let $q$ be a fixed natural number. The set $\ps(\P,S,R)$ can be constructed in polynomial time, given a partial profile $\P$, a sequence $S$ of $q$ candidates and a sequence $R$ of $q$ positions. }
\begin{lemma}
\label{lemma:bucklinDynamic}
\lemmabucklinDynamic
\end{lemma}
\begin{proof}
First, for every $i \in [n]$, construct $\ps (P_i, S, R)$ by applying Lemma~\ref{lemma:bucklinScheduling}. Then, given  $\ps ((P_1, \dots, P_i), S, R)$, observe that
$$ \ps ((P_1, \dots, P_{i+1}), S, R) = \set{ \Vec{u} + \Vec{w} : \Vec{u} \in \ps ((P_1, \dots, P_i), S, R), \Vec{w} \in \ps (P_{i+1}, S, R)} $$
where $\Vec{u} + \Vec{w}$ is a point-wise sum of the two vectors $(\Vec{u} + \Vec{w})(j) = \Vec{u}(j)+ \Vec{w}(j)$. Hence, $\ps(\P,S,R)$ can be constructed in polynomial time.
\end{proof}

Finally, we show that we can solve $\decision{\maxp}{{>} k}$ with respect to Bucklin in polynomial time, using Lemma~\ref{lemma:bucklinDynamic}.

\begin{theorem} 
\label{thm:bucklinMaxBounded}
For all fixed $k \geq 1$, $\decision{\maxp}{{>}k}$ is solvable in polynomial time under the Bucklin rule.
\end{theorem}
\begin{proof}
To determine whether $\maxrnk c\P\tie > k$, we seek $k$ candidates that defeat $c$ in some completion. For that,
we consider each subset
$\set{c_1, \dots, c_k} \subseteq C \setminus \set{c}$ and determine
whether these $k+1$ candidates can be placed in positions so that $c_1, \dots, c_k$ all defeat $c$.

There are two cases where $c_i$ defeats $c$ in a completion $\T$. If $c_i \tie c$ then we need
$s(\T, c_i) \leq s(\T, c)$, that is, there should be $t \in [m]$ such that $c$ is among the top $t-1$ votes for at most $n/2$ votes (hence $s(\T, c) \geq t$) and $c_i$ is among the top $t$ for more than $n/2$ votes (hence $s(\T, c_i) \leq t$). Otherwise, if $c \tie c_i$, then we need $s(\T, c_i) < s(\T, c)$ and we change the condition for $c_i$: it should be among the top $t-1$ for more than $n/2$ votes (hence $s(\T, c_i) \leq t-1 < s(\T, c)$).

Given $c_1, \dots, c_k \in C \setminus \set{c}$ and $t \in [m]$, define a sequence $S = (c_1, \dots, c_k, c)$  of candidates
and $R = (r_1, \dots, r_k, t)$ of positions: if $c_i \tie c$ then $r_i = t$, otherwise $r_i = t-1$. Construct $\ps(\P, S, R)$ via
Lemma~\ref{lemma:bucklinDynamic}, and test whether there is $(h_1, \dots, h_{k+1}) \in \ps(\P, S, R)$ such that $h_{k+1} \leq n/2$ and $h_i > n/2$ for every $i \in [k]$.
\end{proof}

\subsection{Maximin}
Let $N_\T(c, c')$ be the number of votes that rank $c$ ahead of $c'$ in the profile $\T$. The score of $c$ is $s(\T, c) = \min \set{N_\T(c, c') : c' \in C \setminus \set{c}}$.  A winner is a candidate with a maximal score. Xia and Conitzer~\cite{DBLP:journals/jair/XiaC11} established that under Maximin, $\PW$ is NP-complete; we generalize it and show that $\decision{\minp}{{<k}}$ is NP-complete for every $k>1$.  They also show that $\NW$ is tractable, and their polynomial-time algorithm can be easily adjusted to solve $\decision{\maxp}{{>1}}$ by accommodating tie-breaking. In contrast, we show  that $\decision{\maxp}{{>k}}$ is NP-complete for every $k > 1$. 

For the hardness results, we use the following technique. For a profile $\T$ and a pair of candidates $c, c'$ define the \e{pairwise score difference} $D_\T(c,c') = N_\T(c, c') - N_\T(c', c)$. Note that $D_\T(c,c') = - D_\T(c',c)$ and $D_\T(c,c') = 2N_\T(c, c') - n$, so we can define the score under Maximin to be 
$s(\T, c) = \min \set{D_\T(c, c') : c' \in C \setminus \set{c}}$.
The following lemma states that we can change the values of $D_\T$ to any other values, as long as the parity of the values is unchanged.

\begin{lemma}[Main theorem in~\cite{article:Mcgarvey}]
  \label{lemma:pairwiseScores}
Let $\T$ be a profile and $F \colon C \times C \rightarrow \integer$ be a skew-symmetric function (i.e., $F(c_1, c_2) = -F(c_2, c_1)$) such that for all pairs $c, c' \in C$ of candidates, $F(c,c') - D_\T(c, c')$ is even. There exists a profile $\T'$ such that $D_{\T \circ \T'} = F$ and $|\T'| \leq \frac{1}{2} \sum_{c, c'} (|F(c, c') - D_\T(c, c')| + 1)$.
\end{lemma}

If the values $|F(c, c') - D_\T(c, c')|$ are polynomial in $n$ and $m$, then $\T'$ of Lemma~\ref{lemma:pairwiseScores} can be constructed in polynomial time. This is used for establishing the following results.

\def\thmmaximinMinBounded{ Under Maximin, $\decision{\minp}{{<}k}$ is NP-complete for all fixed $k \geq 2$. }
\begin{theorem}
\label{thm:maximinMinBounded}
\thmmaximinMinBounded
\end{theorem}
\begin{proof}
Let $k \geq 1$. We show a reduction from \PW to $\decision{\minp}{{<}k+1}$ under Maximin. Let $\P = (P_1, \dots, P_n)$ and $c^*$ be an input for \PW over a set $C$ of $m$ candidates. By the proof of Xia and Conitzer~\cite{DBLP:journals/jair/XiaC11} that \PW is hard for Maximin, we can assume that for every completion $\T$ of $\P$ the score of $c^*$ satisfies $s(\T, c^*) \leq -2$. As in the proof of Theorem~\ref{thm:PolyRulePosMemk}, the idea is to add $k-1$ new candidates and modify the voters so that the new candidates are always the top $k-1$ candidates, and the score of every original candidate is increased by the same amount.

We define $C', \P'$ and $\tie$ as follows. The candidate set is $C' = C \cup D$ where $D = \set{d_1, \dots, d_{k-1}}$ and the tiebreaker is $\tie = O(D, \set{c^*}, C \setminus \set{c^*})$. The profile $\P' = \P_1 \circ \T_2$ is the concatenation of two parts. The first part is $\P_1 = (P_1', \dots, P'_n)$ where $P_i'$ is the same as $P_i$, except that the candidates of $D$ are placed at the bottom positions. Formally, $P_i' \eqdef P_i \cup P(C, d_1, \dots, d_{k-1})$. Observe that for every $c \in C$ and $d \in D$, the pairwise score difference $D_{\T_1}(c, d)$ is the same in every completions $\T_1$ of $\P_1$. The same holds for $D_{\T_1}(d, d')$ on all $d, d' \in D$.

The second part, $\T_2$, is the complete profile that exists due to Lemma~\ref{lemma:pairwiseScores} such that for every completion $\T' = \T_1 \circ \T_2$ of $\P'$, the pairwise scores differences satisfy: \begin{itemize}
    \item $D_{\T'}(d, c) \in \set{-1, 0, 1}$ for all $d \in D$ and $c \in C' \setminus \set{d}$;
    \item $D_{\T'}(c,c')  = D_{\T_1}(c,c')$  for all $c, c' \in C$.
\end{itemize}

We show that $c^*$ is a possible winner of $\P$ if and only if $\minrnk c {\P'} \tie < k+1$. Let $\T = (T_1, \dots, T_n)$ be a completion of $\P$ where $c^*$ is a winner, for every $c \in C$ we must have $s(\T, c) \leq s(\T, c^*) \leq -2$. Define a completion $\T' = \T_1 \circ \T_2$, $\T_1 = (T_1', \dots, T_n')$ of $\P'$ where $T_i' = T_i \circ (d_1, \dots, d_{k-1})$. For every $d \in D$ the score satisfies $s(\T', d) \geq -1$ by the definition of $\T_2$. For every $c \in C$, for every $d \in D$ we have $D_{\T'}(c, d) \geq -1$ and for every $c' \in C \setminus \set{c}$ we have $D_{\T'}(c,c')  = D_{\T_1}(c,c') = D_{\T}(c,c')$. Since $s(\T, c) \leq -2$, we can deduce that
$$ s(\T', c) = \min \set{D_{\T'}(c, c') : c' \in C' \setminus \set{c}} = \min \set{D_{\T}(c, c') : c' \in C \setminus \set{c}} = s(\T, c)\,. $$
Overall, for every $c \in C$ we have $s(\T', c) = s(\T, c)$ and $s(\T', c) < s(\T',d )$ for every $d \in D$. By the definition of $\tie$, we get that only the candidates of $D$ defeat $c^*$ in $\T'$, therefore $\rnk {c^*} {\T'} \tie = k$.

Conversely, let $\T' = \T_1 \circ \T_2$ be a completion of $\P'$ where $\rnk {c^*} {\T'} \tie \leq k$, we know that $s(\T', c^*) \leq -2$ and $s(\T', d) \geq -1$ for every $d \in D$. Since the $k-1$ candidates of $D$ defeat $c^*$ in $\T'$, we can deduce that $c^*$ defeats all candidates of $C$ in $\T'$. Define a completion $\T$ of $\P$ be removing the candidates of $D$ from the voters of $\T_1$. For every $c \in C$ we have $s(\T', c) \leq -2$ and $D_{\T'}(c, d) \geq -1$ for every $d \in D$, hence
$$ s(\T, c) = \min \set{D_{\T}(c, c') : c' \in C \setminus \set{c}} = \min \set{D_{\T'}(c, c') : c' \in C' \setminus \set{c}} = s(\T', c)\,. $$
We can deduce that $c^*$ is a winner in $\T$.
\end{proof}

\def\thmmaximinMaxBounded{
Under Maximin, $\decision{\maxp}{{>}k}$ is
solvable in polynomial time for $k=1$ and is NP-complete for all $k>1$. }
\begin{theorem}
\label{thm:maximinMaxBounded}
\thmmaximinMaxBounded
\end{theorem}
\begin{proof}
As said earlier, tractability for $k=1$ is obtained by adjusting the $\NW$ algorithm of Xia and Conitzer~\cite{DBLP:journals/jair/XiaC11}. For $k>1$, we show a reduction from \e{exact cover by-3sets} (X3C): given a vertex set $U = \set{u_1, \dots, u_{3q}}$ and a collection $E = \set{e_1, \dots, e_m}$ of 3-element subsets of $U$, can we cover all the elements of $U$ using $q$ pairwise-disjoint sets from $E$? This problem is known to be NP-complete~\cite{DBLP:books/fm/GareyJ79}.

Given $U$ and $E$, we construct an instance $(C, \P, \tie)$ under Maximin. The candidate set is $C = U \cup \set{c^*, w} \cup D$ where $D = \set{d_1, \dots, d_k}$, and the tiebreaker is $\tie = O(D, \set{c^*}, U \cup \set{w})$. The voting profile $\P = \P_1 \circ \T_2$ is the concatenation of two parts that we describe next.

The first part $\P_1 = \set{P_e}_{e \in E}$ contains a voter for every set in $E$. For every $e \in E$, define a complete order 
$T_e = O(w, c^*, U \setminus e, e, d_1, \dots, d_k) $.
The partial order $P_e$ is obtained from $T_e$ by removing the relations in $(e \cup \set{d_1, \dots, d_{k-1}}) \times \set{d_k}$. Denote $\T_1 = \set{T_e}_{e \in E}$. The idea is that ranking $d_k$ higher than the candidates of $e$ indicates that $e$ is in the cover, and ranking $d_k$ in the last position indicates that $e$ is not in the cover. The second part $\T_2$ is the profile that exists due to Lemma~\ref{lemma:pairwiseScores} such that the pairwise scores differences of $\T = \T_1 \circ \T_2$ satisfy:
\begin{itemize}
    \item $D_\T(w, c^*) = m$, $D_\T(w, d_1) = -m-2$, and $D_\T(w, u) = m+2$  for all $u \in U$.
    \item $D_\T(d_k, d_i) = 2q-m$ for $i < k$, and $D_\T(d_k, u) = -m-2$ for all $u \in U$.
    \item $D_\T(c_1, c_2) \in \set{-1, 0,1}$ for every other pair $c_1, c_2 \in C$.
\end{itemize}

Note that $\T^2$ can be constructed in $\poly(m, q)$ time. Next, we state some observations regarding the profile. Let $\T' = \T_1' \circ \T_2$ be a completion of $\P$. For $c^*$ we have $D_\T(c^*, w) = -m$ and $D_\T(c^*, c) \geq -1$ for every other candidate $c$, hence $s(\T, c^*) = -m$. For every $u \in U$, we have $D_\T(u, w) = -m-2$ hence $s(\T, u) \leq -m-2$, and similarly $s(\T, w) \leq -m-2$. We can deduce that $c^*$ always defeats the candidates of $U \cup \set{w}$.  We show that there is an exact cover if and only if $\maxrnk {c^*}\P\tie > k$, which completes the proof of hardness for $\decision{\maxp}{{>}k}$.

Let $B \subseteq E$ be an exact cover, in particular $|B| = q$. Define a completion
$\T' = \T_1' \circ \T_2$, where $\T_1' = \set{T'_e}_{e \in E}$, as follows. For every $e \in E$, if $e \in B$ then define
\[ T'_e = O(w, c^*, U \setminus e, d_k, e, d_1, \dots, d_{k-1}) \]
(i.e., add $d_k \succ (e \cup \set{d_1, \dots, d_{k-1}})$ to $P_e$), otherwise $T'_e = T_e$. We know that $s(\T, c^*) = -m$ and for every $u \in U$, $s(\T, u) = s(\T, w) = -m-2$, so we need compute the scores of the candidates of $D$. For every $i < k$, observe that $d_k$ is raised higher than $d_i$ in the $q$ voters of $\T_1$ which correspond to the edges of $B$, hence $D_{\T'}(d_i, d_k) = D_\T(d_i, d_k) -2q = -m$. For every other candidate $c$ we have $D_{\T'} (d_i, c) \geq -1$ hence $s(\T', d_1) = -m$. 

For $d_k$, for every $u \in U$ there exists a single edge $e \in B$ such that $u \in e$. In the voter $T'_e$, $d_k$ is raised higher than $u$, hence $D_{\T'}(d_k, u) = D_\T(d_k, u) + 2 = -m$. For every other candidate $c$ we have $D_{\T'} (d_k, c) \geq -1$ hence $s(\T', d_k) = -m$. Overall, by the definition of $\tie$, all candidates of $D$ defeat $c^*$, thus $\rnk {c^*} {\T'} \tie > k$.

Conversely, assume that $\maxrnk {c^*}\P\tie > k$, let $\T' = \T_1' \circ \T_2$, $\T_1' = \set{T'_e}_{e \in E}$ be a completion where $\rnk {c^*} {\T'} \tie > k$. As we said $c^*$ always defeats the candidates of $U \cup \set{w}$, hence all candidates of $D$ defeat $c^*$ in $\T'$. Let $B \subseteq E$ be the set of edges $e$ for which $d_k$ is raised higher than $d_1$ in $T'_e$. If $|B| \geq q+1$ then 
\[ D_{\T'}(d_1, d_k) \leq D_\T(d_1, d_k) - 2(q+1) \leq -m-2 \]
hence $s(\T', d_1) \leq m-2 < s(\T', c^*)$, that is a contradiction. From now we assume $|B| \leq q$. Let $u \in E$, assume to the contrary that $u$ is ranked above $d_k$ in all voters of $\T_1'$. In this case we have $D_{\T'} (d_k, u) = D_\T (d_k, u) = -m-2$, which implies $s(\T', d_k) \leq -m-2 < s(\T', c^*)$, that is a contradiction. Hence there exists $e \in E$ such that $d_k$ is raised higher than $u$ in $T'_e$. By the construction of $\P_1$, we get that $u \in e$ and $e \in B$ (if $d_k$ is raised higher than $u$, then it is also raised higher than $d_1$). Overall, $|B| \leq q$ and the edges of $B$ cover $U$, therefore $B$ is an exact cover.
\end{proof}

\section{Concluding Remarks}
We studied the problems of determining the minimal and maximal ranks of a candidate in a partial voting profile, for positional scoring rules and for several other voting rules that are based on scores (namely Bucklin, Copeland and Maximin). We showed that these problems are fundamentally harder than the necessary and possible winners that reason about being top ranked. For example, comparing the maximal/minimal rank to a given number is NP-hard for \e{every} positional scoring rule, pure or not, including plurality and veto.  For the problems of comparison to a fixed $k$, we have generally recovered the tractable cases of the necessary winners (for maximum rank) and possible winners (for minimum rank). An exception is the Maximin rule, where the problem is tractable for $k=1$ but intractable for every $k>1$. Many problems are left for investigation in future research, including: \e{(a)} establishing useful tractability conditions for an input $k$; \e{(b)} completing a full classification of the class of (pure) positional scoring rules for fixed $k$; \e{and (c)} determining the parameterized complexity of the problems when $k$ is the parameter.

\bibliography{References}

\end{document}